\documentclass[aps,10pt,amsmath,amssymb,nofootinbib,notitlepage,tightenlines,longbibliography,twocolumn,superscriptaddress]{revtex4-2}
\pdfoutput=1

\usepackage[dvipsnames]{xcolor}
\usepackage[normalem]{ulem}

\usepackage{amsmath,amsthm,amsfonts,amssymb}
\usepackage{bm}
\usepackage{mathrsfs,soul}
\usepackage{enumerate}

\usepackage{graphicx}
\usepackage[caption=false]{subfig}

\usepackage[breakable,skins]{tcolorbox}
\usepackage{dcolumn}
\usepackage{chngcntr}
\usepackage{enumitem}
\usepackage[ruled]{algorithm2e}

\usepackage{hyperref}
\def\fref#1{Figure~\ref{#1}}
\def\eqref#1{Eq.~(\ref{#1})}
\newcommand{\sm}{Appendix}
\newcommand{\mmsec}{Methods}

\usepackage{hyphenat}

\newcommand{\inlineheading}[1]{\textbf{#1~-- }\ignorespaces}
\newcommand{\inlinesubheading}[1]{\textit{#1~-- }\ignorespaces}

\def\ket#1{|#1\rangle}
\def\bra#1{\langle#1|}

\newcommand{\past}[1]{\overleftarrow{#1}}
\newcommand{\future}[1]{\overrightarrow{#1}}

\newcommand{\ints}{\mathbb{Z}}

\newcommand{\encod}{\mathcal{E}}
\newcommand{\morph}{\mathcal{F}}
\newcommand{\mr}{\kappa}

\newtheorem{theorem}{Theorem}
\newtheorem{lemma}{Lemma}

\newtheorem{definition}{Definition}

\newcommand{\caphead}[1]{\textbf{#1}}

\usepackage{times}

\makeatletter
\newcommand*{\balancecolsandclearpage}{%
  \close@column@grid
  \clearpage
  \twocolumngrid
}
\makeatother

\begin{document}

\title{Provable superior accuracy in machine learned quantum models}

\newcommand{\spms}{Nanyang Quantum Hub, School of Physical and Mathematical Sciences,\\ Nanyang Technological University, Singapore 637371}
\newcommand{\cinst}{Complexity Institute, Nanyang Technological University, Singapore 637335}
\newcommand{\iqoqi}{Institute for Quantum Optics and Quantum Information,\\Austrian Academy of Sciences, Boltzmanngasse 3, Vienna 1090, Austria}
\newcommand{\sustc}{Shenzhen Institute for Quantum Science and Engineering and Department~of~Physics,\\Southern University of Science and Technology, Shenzhen 518055, China}
\newcommand{\cqcct}{Centre for Quantum Computation and Communication Technology, Australian Research Council,\\ Centre for Quantum Dynamics, Griffith University, Brisbane, Queensland 4111, Australia}
\newcommand{\cqt}{Centre for Quantum Technologies, National University of Singapore, 3 Science Drive 2, Singapore 117543}
\newcommand{\dahlem}{Dahlem Center for Complex Quantum Systems, Freie Universit\"at Berlin, 14195 Berlin, Germany}
\newcommand{\wolfsen}{Wolfson College, University of Oxford, Oxford, OX2 6UD, United Kingdom}

\author{Chengran Yang}
\email{cr.yang@nus.edu.sg}
\affiliation{\spms}
\affiliation{\cinst}
\affiliation{\cqt}
\author{Andrew~J.~P.~Garner}
\affiliation{\spms}
\affiliation{\cinst}
\affiliation{\iqoqi}
\author{Feiyang Liu}
\affiliation{\sustc}
\author{Nora~Tischler}
\affiliation{\cqcct}
\affiliation{\dahlem}
\author{Jayne Thompson}
\affiliation{\cqt}
\author{Man-Hong Yung}
\affiliation{\sustc}
\author{Mile Gu}
\email{mgu@quantumcomplexity.org}
\affiliation{\spms}
\affiliation{\cinst}
\affiliation{\cqt}
\author{Oscar Dahlsten}
\email{dahlsten@sustech.edu.cn}
\affiliation{\sustc}
\affiliation{\wolfsen}

\date{\today}

\begin{abstract} 
In modelling complex processes, the potential past data that influence future expectations are immense. 
Models that track all this data are not only computationally wasteful but also shed little light on what past data most influence the future. 
There is thus enormous interest in dimensional reduction – finding automated means to reduce the memory dimension of our models while minimizing its impact on its predictive accuracy. 
Here we construct dimensionally reduced quantum models by machine learning methods that can achieve greater accuracy than provably optimal classical counterparts. 
We demonstrate this advantage on present-day quantum computing hardware. 
Our algorithm works directly off classical time-series data and can thus be deployed in real-world settings. 
These techniques illustrate the immediate relevance of quantum technologies to time-series analysis and offer a rare instance where the resulting quantum advantage can be provably established.
\end{abstract}

\maketitle

\section*{Introduction}
Predictive modelling permeates quantitative science, from the rise and fall of stock markets to the dynamics of neural spike trains. 
The making of quantitative predictions often relies on learning the potential causal dynamics that underlie vast amounts of  observational data. In this vein, there is a significant interest to explain more with less -- enabling us to pinpoint key observations from the past that have maximal influence on future behaviour.
This motivates efficient models, which generate faithful future predictions while retaining only a small fraction of information from the past. 
The benefits of such models are two-fold. 
Firstly, they help us minimize the memory we need to make predictions about the future. 
Secondly, they enable us to sample a wide array of possible futures by adjusting a few key initial parameters. 
Indeed, much of machine learning, from dimensional reduction and feature detection to auto-encoders, adheres to such philosophy~\cite{van2009dimensionality,alpaydin2020introduction,ng2011sparse}.

\begin{figure}[tbp]
\centering
	\includegraphics[width=0.8\linewidth]{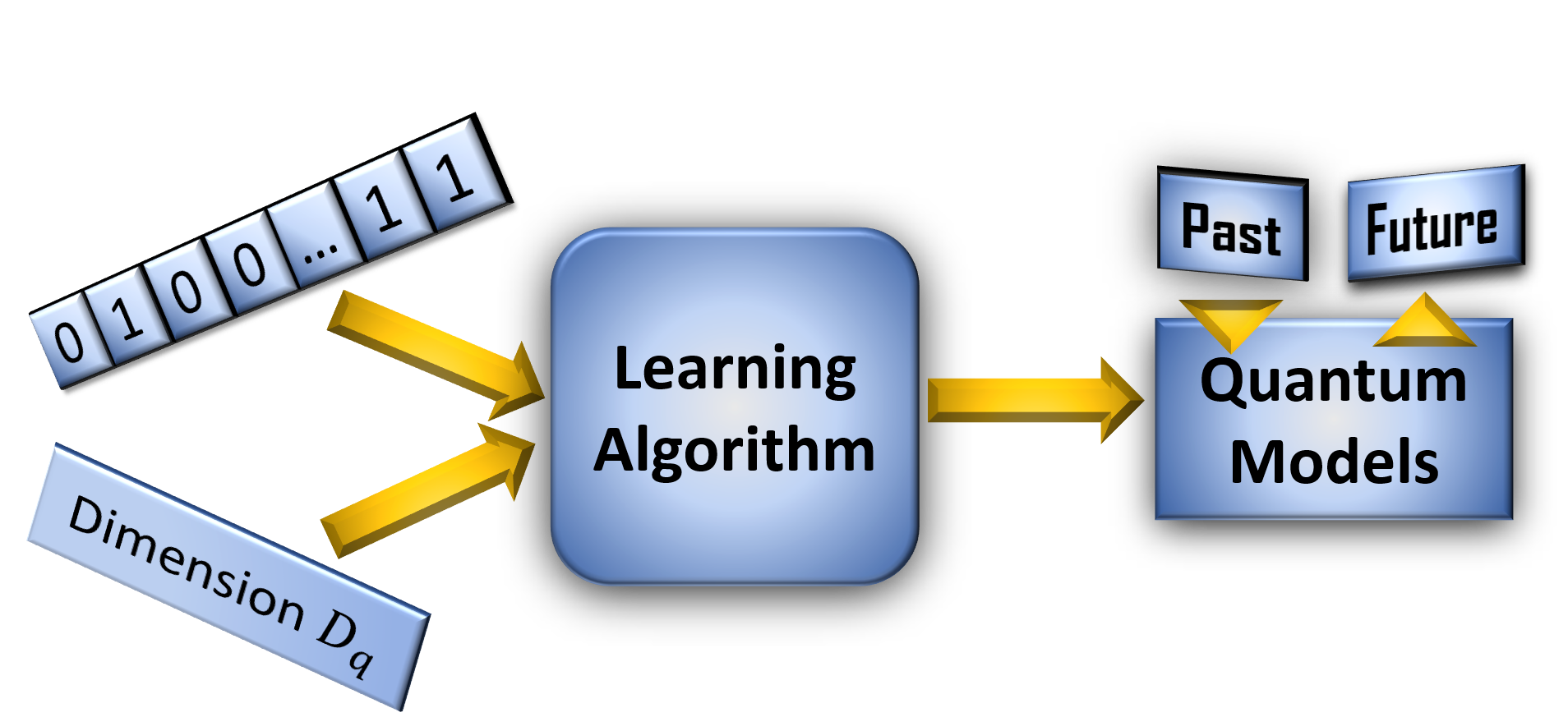}
	\caption{
	\caphead{Scheme for our quantum causal-inference algorithm.}
The algorithm takes as input time-series data drawn from some stochastic process, together with a specification of the number of dimensions of available quantum memory (the number of mutually distinguishable quantum states).
The algorithm identifies a quantum machine that best predicts future outcomes of the process, subject to such memory constraints.
The resulting machines can achieve significantly greater accuracy than classical counterparts.}
	\label{Fig: illustration}
\end{figure}

Even efficient models can be highly memory intensive. Complexity scientists in computational mechanics, for example, have extensively studied the inference of predictive models for stochastic processes~\cite{pearl2000models, Shalizi2001, shalizi2002algorithm, shalizi2014blind}. These processes  govern the behaviour of a system observed at discrete points in time. The goal is then to infer a systematic method of encoding past observations into some machine, whose dynamics allow it to sequentially generate statistically faithful conditional future behaviour at all future time-steps. Sophisticated techniques have been developed for realistic settings where predictive models with memory constraints are trained on finite time-series data -- techniques that have successfully inferred structure in diverse contexts, from distinguishing states of brain activity to causal inference in geomagnetic time-series~\cite{Haslinger2009,Clarke2003,munoz2020general}. Meanwhile, theory indicated that as processes become progressively non-Markovian, the number of distinct configurations a machine needs to store relevant past information can grow without bound~\cite{Elliott2018Continous, elliott2019extreme}. Thus, memory constraints on such models invariably lead to distortions (i.e.\ statistical errors) in predicted futures.

\clearpage
Here we ask: \emph{\ Can memory-limited quantum models -- quantum machines armed with quantum memories -- reduce such distortions?} To answer this, we develop a quantum model discovery algorithm that infers such models directly from raw data (see \fref{Fig: illustration}). Applying this algorithm, we show that our inferred quantum models can achieve accuracies in prediction impossible in provably optimal classical limits and illustrate that these advantages can grow with the non-Markovianity of the processes under study. These results complement recent studies showing how quantum mechanics can reduce memory requirements for the idealistic limiting case of exact modelling~\cite{gu2012quantum,Aghamohammadi2017PRX,Binder2017,yang2018matrix, Ghafari19a,liu2019optimal,Thompson2017,Elliott2018Continous}, which left an open question on whether quantum advantages prevail in more realistic memory-limited scenarios where there is a trade-off between memory and accuracy. Our results prove that the answer is in the affirmative, and additionally give an algorithmic means to infer such dimensionality-reduced quantum models directly from experimental data.
We implement a model found by our algorithm on the IBM  quantum computer (``ibmq athens''). Despite the noise, a statistically significant accuracy advantage remains.

\section*{Results}

\begin{figure}[thbp]
	\centering
	\includegraphics[width=0.85\linewidth]{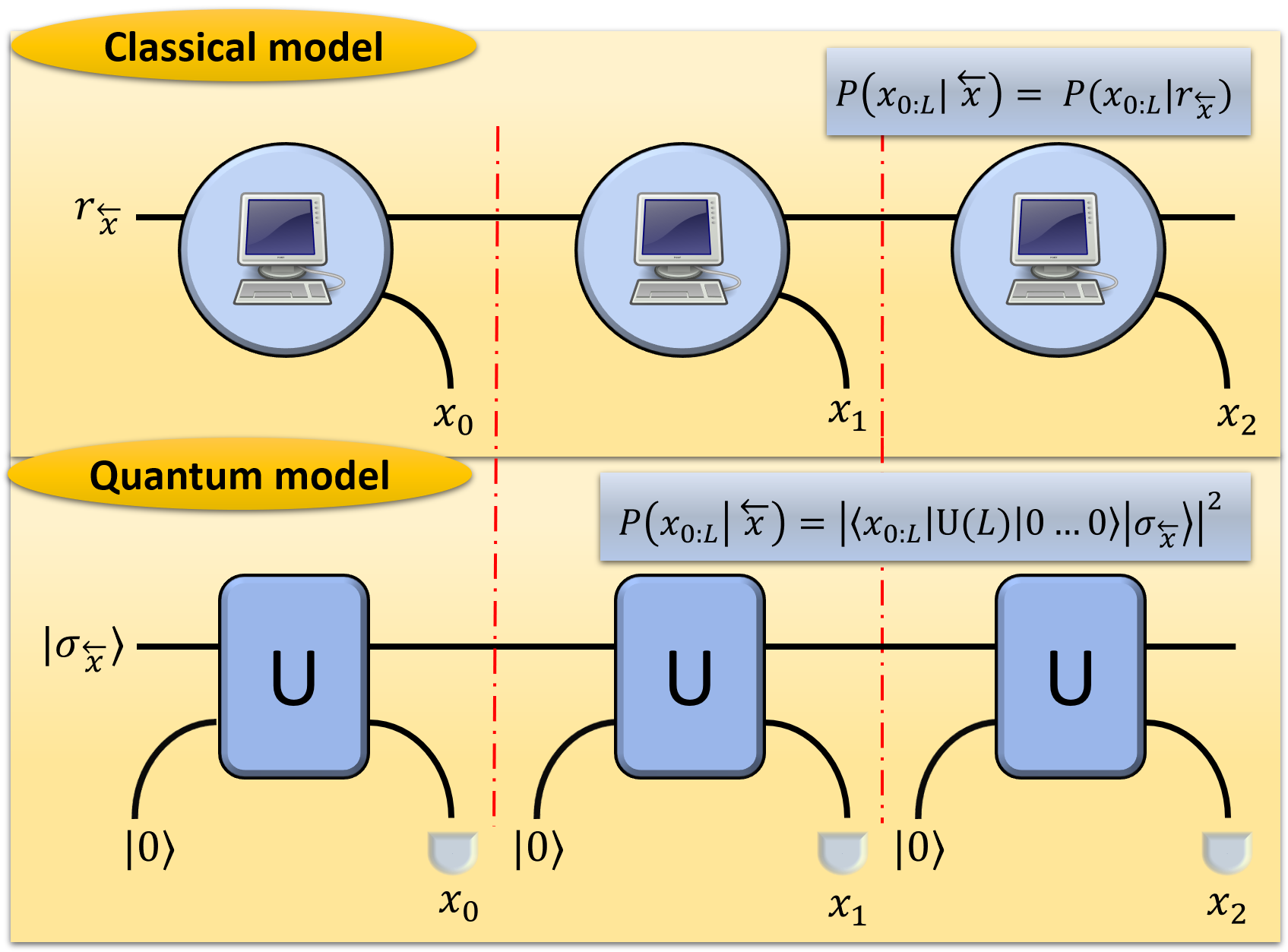}
	\caption{\caphead{Comparison of classical and quantum models} that produce statistical samples of a stochastic process's future behaviour.	A classical model stores classical states in its memory.	In quantum models, the quantum memory and the output register -- initialized in state $\ket{0}$ -- interact through the unitary transformation $U$, and a measurement of the output register in the computational basis yields $x_i$. 	This can be repeated to sample $x_{0:L}$ for as big $L$ as desired. $U(L)$ denotes the unitary coupling the memory with $L$ consecutive output registers.}
	\label{fig:UnitaryCircuit}
\end{figure}
\inlineheading{Classical models and distortion}
Consider a physical system observed at discrete times $t$ with corresponding outcomes $x_t\in\mathcal{X}$ over some finite alphabet $\mathcal{X}$. We assume that the system is stationary and stochastic, such that each $x_t$ is drawn from random variable $X_t$, all of which are governed by a time-translation invariant distribution $P(\past{X},\future{X})$. Here, $\past{X}:= \cdots X_{-2}X_{-1}$ and $\future{X}:=  X_{0}X_{1}X_{2}\cdots$ denote random variables governing stochastic outcomes in the past and future.

An (exact) {\em predictive model} is a systematic algorithm that takes each possible past $\past{x}$,
 encodes it via the deterministic map $\encod$ to some suitable state $\encod(\past{x})$ within a physical memory $M$,
 such that the same systematic action $\mathcal{P}$ on $M$ at each subsequent time-step sequentially outputs $x_0,x_1,x_2 \cdots$
 with probability $P(\future{X} = \future{x}|\past{x})$, as shown in ~\fref{fig:UnitaryCircuit}.  Knowing the state of the memory $M$ is thus as useful as $\past{x}$ for purposes of inferring future statistics. Together, these properties moreover imply that the model is \emph{unifilar} -- if a predictive model in state $r_{\past{x}} = \encod(\past{x})$ emits some output $x_0$, we know with certainty that the model will transition to state  $r_{\past{x},x_0} = \encod(\past{x}, x_0)$~\cite{Shalizi2001}. This guarantees that regardless of how many steps $L$ of the process are produced, it is always possible to work out what state we should update $M$ to, such that the model can continue the simulation process.

The classical predictive models of provably minimal memory dimension are {\em $\varepsilon$-machines}~\cite{Shalizi2001,Crutchfield1989}. $\varepsilon$-machines have been extensively studied in the context of complex systems~\cite{Haslinger2009,yang2008increasing,park2007complexity,Boyd2016identifying, Garner2017thermodynamics}. They capitalize on the fact that if two pasts, $\past{x}$ and $\past{x}'$ have coinciding future statistics, then it is unnecessary to distinguish them for purposes of future prediction. This then divides the set of all pasts into equivalence classes, according to the relation $\past{x} \sim \past{x}'$ if and only if  $P(\future{X}|\past{x}) = P(\future{X}|\past{x}^{'})$. The $\varepsilon$-machine then allocates in memory one configuration state for each equivalence class. These states are known as \emph{causal states}, denoted by $\encod(\past{x})$.

The number of causal states, $d_c$, is known as a process's \emph{topological complexity} -- capturing the minimal number of configurations a predictive model requires, whilst  exhibiting correct conditional future behaviour. For Markovian processes, for example, the map $\encod({\past{x}}) = x_{-1}$ is predictive, since knowing the very last output is sufficient for purposes of prediction without distortion. Thus for Markovian processes, $d_c$ needs never exceed the size of the alphabet $\mathcal{X}$. When a process becomes highly non-Markovian, its topological complexity can scale without bound~\cite{marzen2015informational}. The simplest example is a discrete-time clock, which `ticks' by emitting `$1$' every $T$ time-steps and `$0$' at all other times.

In such scenarios of diverging $d_c$, finite memory means a necessary {\em distortion} of predictive accuracy.  Specifically, given the constraint that our memory $M$ has a fixed number of dimensions $\hat{d}_c \leq d_c$, we may define a candidate model as one that encodes $\encod(\past{x})$ within $M$ -- such that systematic action $\mathcal{P}$ on $M$ generates a sample from $\hat{P}\left(\future{X}|\encod(\past{x})\right)$.

Our goal is then to find a model such that the deviation of $\hat{P}$ from $P$-- the distortion--is minimal. There are a number of possible ways to quantify this distortion.
Here, we make use of the normalized conditional Kullback--Leibler (KL) divergence
\begin{equation}
\label{eq:Dfutpast}
	D_L(P, \hat{P} | \past{x}) := \frac{1}{L}\mathcal{D}_{\rm KL}(P(X_{0:L}|\past{x})~||~\hat{P}(X_{0:L} | \encod(\past{x}))),
\end{equation}
where $X_{0:L} = X_0 X_1 \ldots X_{L}$ and $\mathcal{D}_{\rm KL}(P||Q) = \sum_{x} P(x)\log(P(x)/Q(x))$.
The KL divergence holds operational significance in diverse contexts including machine learning, hypothesis testing and thermodynamics~\cite{hobson1987concepts,cover2012elements}.
The {\em distortion} of a model $\hat{P}$ relative to the process $P$ is then quantified as:
\begin{equation}
\label{eq:Distortion}
	D_e(P, \hat{P}) = \lim_{L\to\infty}\sum_{\past{x}} P(\past{x}) \, D_L(P, \hat{P} | \past{x}),
\end{equation}
i.e.\ the mean KL divergence over all pasts.

\inlineheading{Quantum Models}
In conventional models, the memory $M$ is assumed to be classical. Quantum mechanics, however, enables memories configured in superposition states. In this scenario, a $d$-dimensional quantum model is one where each past $\past{x}$ is encoded into a suitable quantum state $\ket{\sigma_{\past{x}}}$ within a $d$- dimensional quantum memory $M$. Meanwhile, there must exist some quantum process whose repeated application allows sampling of the appropriate conditional futures. \fref{fig:UnitaryCircuit} illustrates the quantum circuit picture. At each time-step, $M$ is coupled via a time-step independent unitary $U$ to an output register initially in $\ket{0}$. The output register is then measured, resulting in emitted outcomes $x_1,x_2,\ldots$ governed by $\hat{P}(\future{X}|\sigma_{\past{x}})$.

There are processes where quantum models can achieve zero distortion (i.e. $\hat{P}(\future{X}|\sigma_{\past{x}}) = P(\future{X}|\past{x})$) with a lower memory dimension than classically possible~\cite{Thompson2017}.
Experiments demonstrate that stochastic processes with three causal states can be accurately modelled by encoding the past within a single qubit~\cite{partghafari2019dimensional}. Thus the quantum topological complexity $d_q$ can be strictly lower than its classical counterpart. However, systematic means to construct such models require optimisation over an exponentially growing parameter space~\cite{loomis2019strong,liu2019optimal}.

Zero-distortion is an idealized limit. In practice, we are unlikely to know the exact process $P$ we wish to model. Instead, we are given a finite sequence of observations drawn from this process. In such scenarios, statistical fluctuations imply that aiming for zero distortion is likely meaningless and often represents harmful overfitting. Typically, we start with general finite and fixed memory resources. Thus, minimizing distortion when given finite memory $\hat{d}_q < d_q$ and finite data is of both of conceptual and practical importance. In such settings, the capabilities of quantum models to minimize this distortion are not known. Indeed, the task remains under active study even in classical settings~\cite{tegmark2020pareto,marzen2016predictive,marzen2017nearly}. As such, our goal here is to find an algorithm for inferring quantum models of some given dimension $\hat{d}_q$ that attempts to minimize distortion, and use it to establish that such models exhibit less distortion than any classical counterpart.


\inlineheading{Quantum model discovery algorithm} (see Box~1). The algorithm takes two inputs: (i) a data sequence $x_{0:L}$, assumed to be drawn from some stationary stochastic process $P$ and (ii) our desired memory dimension $\hat{d}_q$.
From this data, the algorithm learns a quantum model of the process, including (i) an encoding map $\encod$ which specifies how to encode each past into a quantum memory of dimension $\hat{d}_q$, and (ii) the relevant physical process -- described by a unitary $U$ as in~\fref{fig:UnitaryCircuit} -- whose repeated application produces  entangled output states. Measurements on such output states emit outputs that approximate the process's conditional future behaviour.

\begin{figure}
\label{Box:Algorithm}
\vspace*{0.5em}
\begin{tcolorbox}[skin=enhanced]
 {\bf Box 1: Quantum model discovery algorithm. }

 \vspace*{0.5em}
 {\bf Inputs:}\\
 \vspace*{0.25em}
 \begin{tabular}{rcl}
	$\hat{d}_q$  &  --  & the desired model memory dimension. \\
	$x_{0:L}$ & -- & a length $L$ sequence from the process.
 \end{tabular}

 \vspace*{0.5em}
 {\bf Outputs:}\\
 \vspace*{0.25em}
 \begin{tabular}{rcl}
	$U$  &  --  & the unitary dynamics of the quantum model. \\
	$\encod$ & -- & the encoding map from observed histories\\&& to $\hat{d}_q$--dimensional quantum memory states.
 \end{tabular}	

 \vspace*{0.5em}
 {\bf Algorithm:}
	\vspace*{-0.5em}
	\begin{enumerate}
		\item To learn the unitary operator $U$:
		\vspace*{-0.5em}
		\begin{enumerate}
		 \item Randomly initialize parameter set $A$ corresponding up to completeness to the set of Kraus operator matrices $\{A_x := \bra{x}U\ket{0}\}$.
 		 \item Evaluate the cost function $C := -\log P(x_{0:L} ; A)$ and its gradient $\nabla C$.
		 \item Update the Kraus operators  $\{A_x \}$ using gradient-descent based method (such as Adam).
		 \item Repeat (b)-(c) until the cost function decrease is sufficiently small.
		 \item Save final $A$ as candidate parameters.
		 \item Repeat (a)-(e) as desired, to minimize the impact of initial choice of $A$. Choose the final $A$ with the lowest C.
		 \item For that $A$, recover the completeness relation for $\{A_x\}$.
		 \item Construct a unitary operator $U$ from these $\{A_x\}$.
		\end{enumerate}
		\item To compute the encoding $\encod$:
		\vspace*{-0.5em}
		\begin{enumerate}
		\item Compute the leading eigenvector $\ket{\sigma_0}$ of one of the Kraus operators ($A_0$).
		\item $\encod$ is then defined as \\$\encod\!\left(\past{x}\right) := A_{\past{x}}\ket{\sigma_0}/\lVert A_{\past{x}}\ket{\sigma_0} \rVert_2$.
		\end{enumerate}
	\end{enumerate}
\end{tcolorbox}
\end{figure}

To operate the learning algorithm, we cast the problem as minimizing some cost function over a parameter set $B$. This involves finding (a) an effective  parametrisation of all such models to optimize over and (b) a computable cost function that proxies the mean KL-divergence $D_e$.

To tackle (a), first note that our model's output behaviour is entirely defined by a family of Kraus operators $A = \{A_x\}$, where $A_{x} = \bra{x}U\ket{0}$ captures the action of the memory when it interacts with an output register that is subsequently observed to have output $x$. Observe also that $A_x$ informs the {\em encoding map} from pasts onto quantum memory states. To see this, note that a memory initially in $\rho$ is mapped to $A_x\rho A_x^\dag$ after observation of $x$. Thus, observation of a past $\past{x}$ implies applying a sequence of Kraus operators $A_{\past{x}} := A_{x_0}A_{x_{-1}}\ldots$. We can therefore encode the $\past{x}$ as
\begin{equation}
	\encod\!\left(\past{x}\right) := \ket{\sigma_{\past{x}}} := c A_{\past{x}}\ket{\sigma_0},
\end{equation}
where $c$ is the constant of normalization. Here, the initial quantum state $\ket{\sigma_0}$ can be the leading eigenvector of any Kraus operator $A_x$ (see \sm~\ref{sem:QEM}).
This na\"ively suggests optimizing over $A$, but this is generally cumbersome as $A$ is constrained by the completeness relation $\sum_x (A_x)^\dagger A_x = \mathbb{I}$.
In the \mmsec{}, we instead devise a way to optimize over a set of unconstrained $\hat{d}_q\times \hat{d}_q$ complex matrices $B = \{B^x\}$, whose value enables the deduction of $A$ via tensor network techniques~\cite{Vidal2007}.

To tackle (b), the obvious candidate cost function is $D_e$ itself. However, computation of this quantity is costly -- especially as we are given only a sample $x_{0:L}$ of the desired process $P$. We thus employ a more efficiently computable proxy,
\begin{equation}
\label{eq:CostFunction}
C(B) = -\log_2 P_B(x_{0:L}),
\end{equation}
representing the negative log-likelihood of producing $x_{0:L}$ with a model parameterized by $B$.
In the limit of $L\to \infty$, a model minimising the negative log-likelihood of \eqref{eq:CostFunction}  exactly replicates the desired stochastic behavior~\cite{petrie1969probabilistic,Juang1985}.
With both the cost function and parameter spaces established, the use of optimisation algorithms such as Adam is enabled (see the \mmsec{}). In the \mmsec{}, we also show that $C$ can be computed directly from $B$ without first deducing $A$, boosting optimisation efficiency. Once optimisation concludes, the corresponding $U$ for generating predictions is retrieved from $B$ (see the \mmsec{}).
\begin{figure}[htbp]
	\centering
	\subfloat[][]{\includegraphics[width=0.7\linewidth]{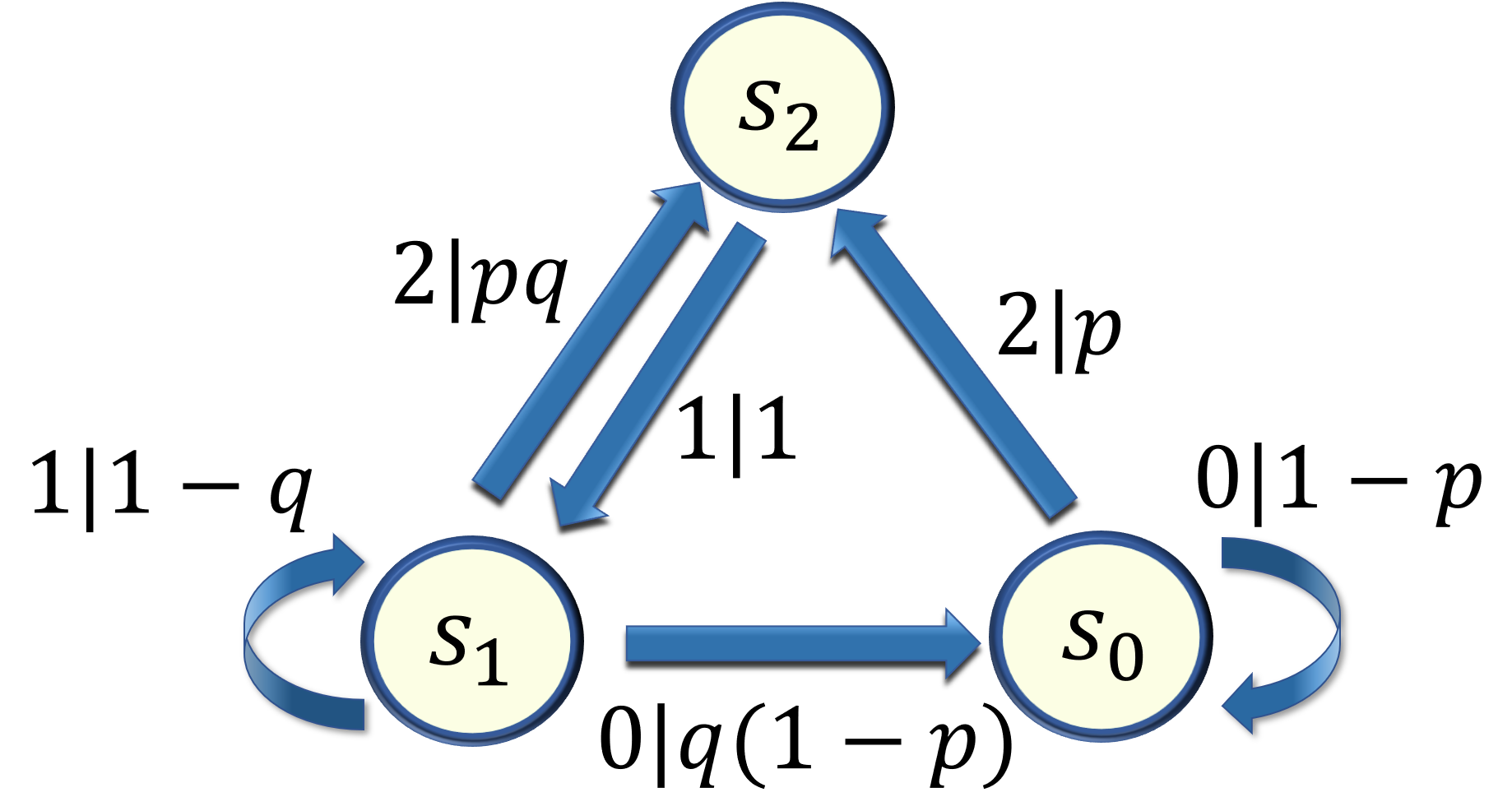}\label{Fig:asymm}}\\
	\subfloat[][]{\includegraphics[width=0.95 \linewidth]{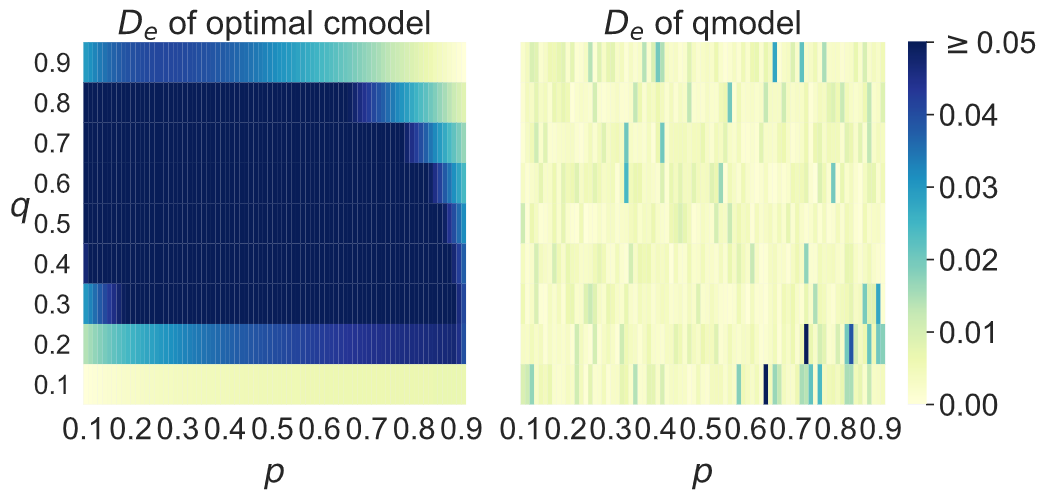}\label{fig:ex1error}}
	\caption{\caphead{Example: The asymmetric process.} \textbf{(a)} $\varepsilon$-machine for the asymmetric process, parametrised by $p,q\in(0,1)$. The nodes are the causal states $s_0$, $s_1$ and $s_2$. The arrows indicate transitions between states, where each label $x|P(x)$ indicates output $x$ is produced with probability~$P(x)$.
	\textbf{(b)} Classical vs quantum distortion.
	Distortion in the $2$-dimensional models for the asymmetric process across a range of process parameters $p$, $q$. The left-hand-side shows the optimal approximate classical model's distortion; the right-hand-side shows the distortion of the most accurate quantum model found by our algorithm.}
\end{figure}

\inlineheading{Provable Quantum Advantage}
We are now in a position to answer the open question: \emph{can memory-constrained quantum models exhibit improved accuracy}? Answering this question in full generality is non-trivial as almost all existing analytical results for predictive modelling apply only to zero distortion. While sophisticated numerical methods for inferring classical models exist~\cite{kelly2012new,shalizi2014blind}, they do not tell us whether the resulting models are provably optimal in minimizing KL-divergence.
In the \mmsec{} we establish provable upper bounds on the quality of classical methods of fixed dimension. We will show that our inferred quantum models can exceed these optimal bounds.

The first example is the family of simple asymmetric processes.
For any process in this family, exact modelling requires three causal states classically (see its optimal predictor in \fref{Fig:asymm}),
 but can be achieved quantumly on a qubit~\cite{Ghafari19a, Thompson2017}. Thus, if we set the desired dimensions to $2$, our quantum model could theoretically achieve a KL-divergence of $0$, while some distortion for classical models is inevitable. Indeed, in the methods, we evaluate the minimal error achieved by $2$-dimensional classical models for process parameters; and compare them with errors in our trained quantum models (see  \fref{fig:ex1error}). We observed that errors in our trained quantum models are negligible in comparison to the best possible classical counterparts.

Our second example involves a family of generalised $3$-state quasi-cycles (see \fref{Fig:quasicycle}). Quasi-cycles are for example of interest in thermodynamics---certain quasi-cycles are valid thermal operations in that they preserve the Gibb's thermal state, but interestingly without respecting detailed balance~\cite{HorodeckiOppenheim2013}. They constitute a family of processes where neither quantum nor classical models can generally achieve perfect accuracy when limited to $2$ dimensional memories~\cite{liu2019optimal}. The resulting error of our inferred quantum model (see~\sm~\ref{sm:learned}) is compared with that of the classical limit in \fref{fig:ex2error}. Clearly, a broad area of process parameters $(p,\delta)$ can be accurately approximated by a compressed quantum model, whose error is up to an order of magnitude better than the classical models. The algorithm also achieves near perfect accuracy over a one-dimensional subspace of the quasi-cycles -- retrieving a class of quantum models that was only recently discovered using exhaustive numerical search~\cite{liu2019optimal}).
\begin{figure}[htbp]
	\centering
	\subfloat[][]{\includegraphics[width=0.6\linewidth]{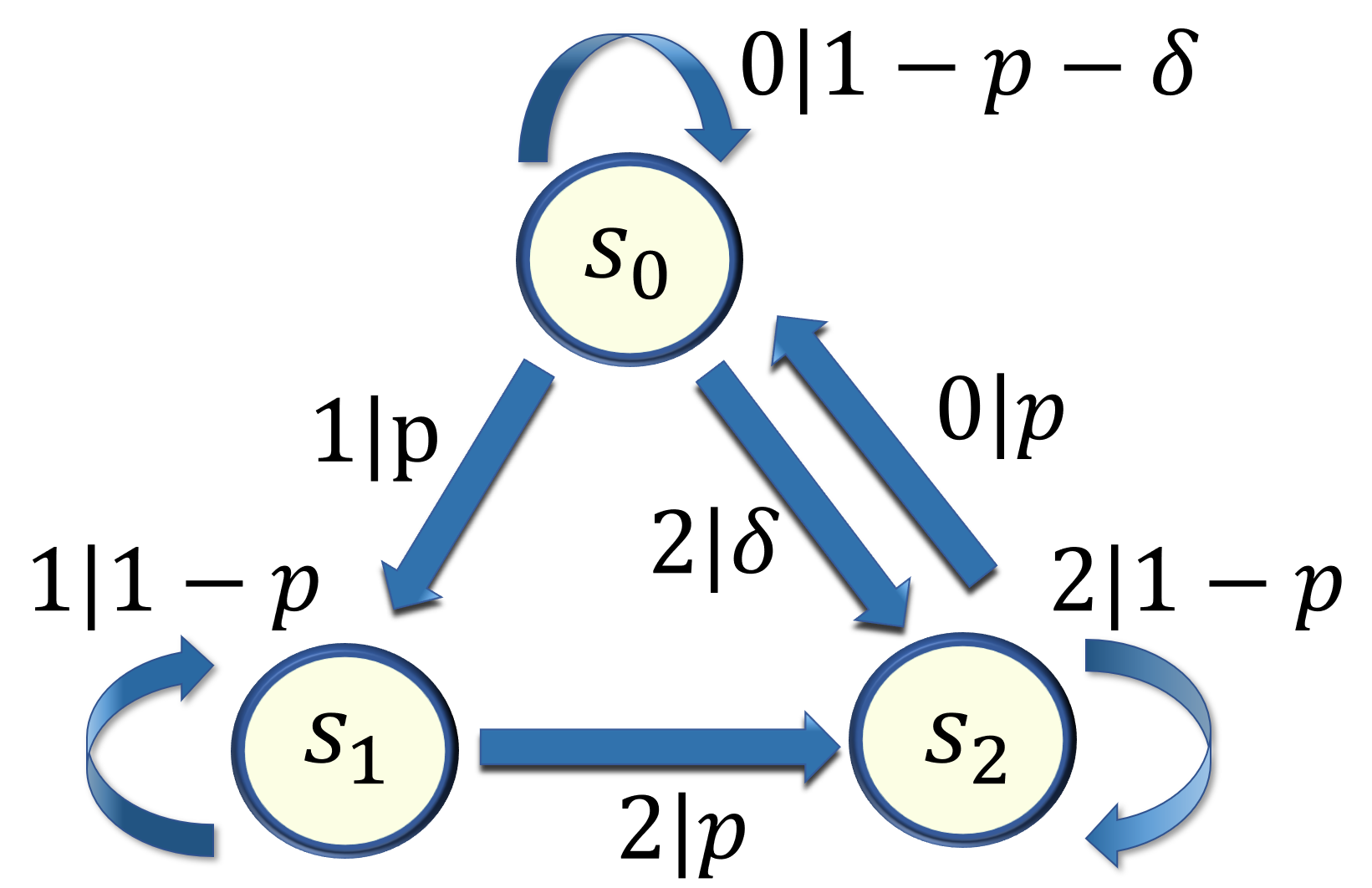}\label{Fig:quasicycle}}\\
	\subfloat[][]{\includegraphics[width=0.95\linewidth]{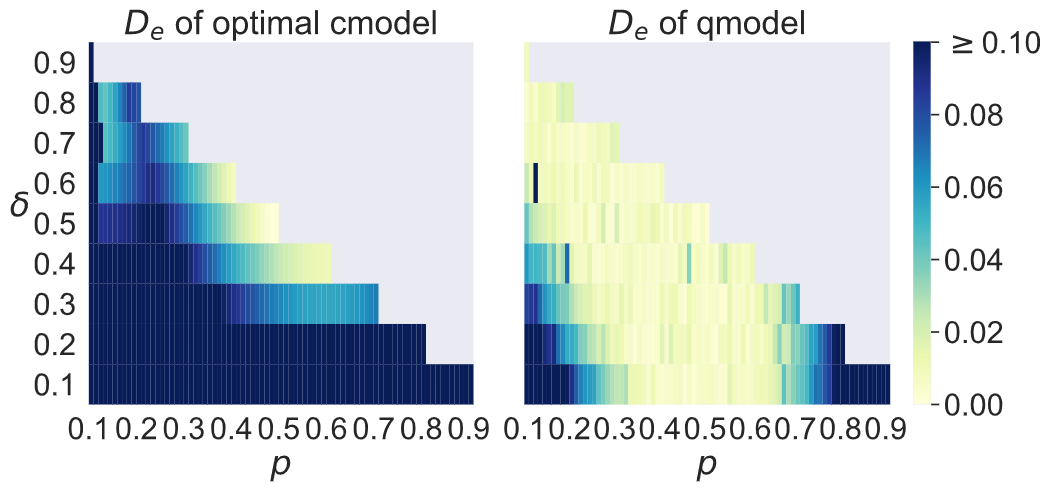}\label{fig:ex2error}}
	\caption{\caphead{Example: Quasi-cycle process.} \textbf{(a)}
	$\varepsilon$-machine for the $3$-state quasicycle, parametrised by $p\in(0,1), \delta\in(0,p]$. The nodes are the causal states $s_0$, $s_1$ and $s_2$.
	The arrows indicate transitions between states, where each label $x|P(x)$ indicates output $x$ is produced with probability~$P(x)$. \textbf{(b)}  Classical vs quantum error comparison.	Error in the $2$-dimensional models for the quasi-cycle, across a range of process parameters $p$, $\delta$. The left-hand-side shows the optimal approximate classical model's error; the right-hand-side shows the error of the most accurate quantum model found by our algorithm.}
\end{figure}

We now apply our inference algorithm to an iconic non-Markovian process - the family of discrete-time renewal processes. Renewal processes represent an interpolation between Poissonian processes and ticking clocks, characterising many natural processes including models of lifetimes~\cite{doob1948renewal} and neural spike trains~\cite{crutchfield2015time}. At each time-step, the process may emit one of two outputs: $0$, representing no tick, and $1$ representing a tick, and their probability of ticking at each time-step depends only on the number of time-steps since their previous tick~\cite{Elliott2018Continous,Marzen2017,elliott2019extreme}. In the event that this probability is constant, we have a Poissonian process whose topological complexity is zero; as its action at each time is independent of past events (and thus considered memory-less).

Here we consider a class of period-$N$ uniform renewal processes $\mathcal{U}_N$. These represent processes where the number of $0$s between each two neighboring ticks is uniformly distributed between $0$ and $N-1$. As $N$ scales, the process becomes progressively more non-Markovian. Indeed, the classical topological complexity for $\mathcal{U}_N$ is $N$, the intuition being that a maximally faithful model would need to store the number of time-steps since the last tick to decide with what probability it should tick at each time-step (see corresponding $\varepsilon$-machine in \fref{fig:UniformRenewalProcess}). Reducing the number of distinct memory configurations to $\hat{d} < N$ would force the model to lose track of this information -- resulting in unavoidable distortion.

To study the magnitude of this distortion, we applied our algorithm to infer memory-constrained models of $\mathcal{U}_N$ for cases where $N$ ranged between 2 and 5 inclusive, for a variety of memory limits $\hat{d}_q < N$.
We then compare the error with a lower bound of optimal classical models with similar memory constraints (see the \mmsec{}). The resulting distortions are illustrated in the upper graph of \fref{fig:ex3error}. Even though the learned quantum models have the additional handicap of being trained on finite-length data, we see that quantum models outperform optimal classical counterparts away from the zero-error regime.  Models with larger $\hat{d}_q$ or $\hat{d}_c$ maintain smaller error but at the cost of a larger memory, and requiring a larger dataset to train. In the case of $N=4$, there is a significant gap between the errors of $\hat{d}_q=3,4$ and the error of $\hat{d}_q=2$. Thus, the algorithm could here be used to identify a sweet spot for the choice of~$\hat{d}_q$: setting $\hat{d}_q=3$ gives a good accuracy for a small dimension.
To examine the resilience of the quantum advantage on present-day noisy devices, we implement the quantum model on one of IBM's quantum systems (``ibmq athens") for the case of $\hat{d}_q=2, N =3$.
Such quantum models still exhibit statistically significant accuracy advantage over the classical bound, as is shown in~\fref{fig:IBMdist}.
\begin{figure}[htbp]
	\centering
	\includegraphics[width=\linewidth]{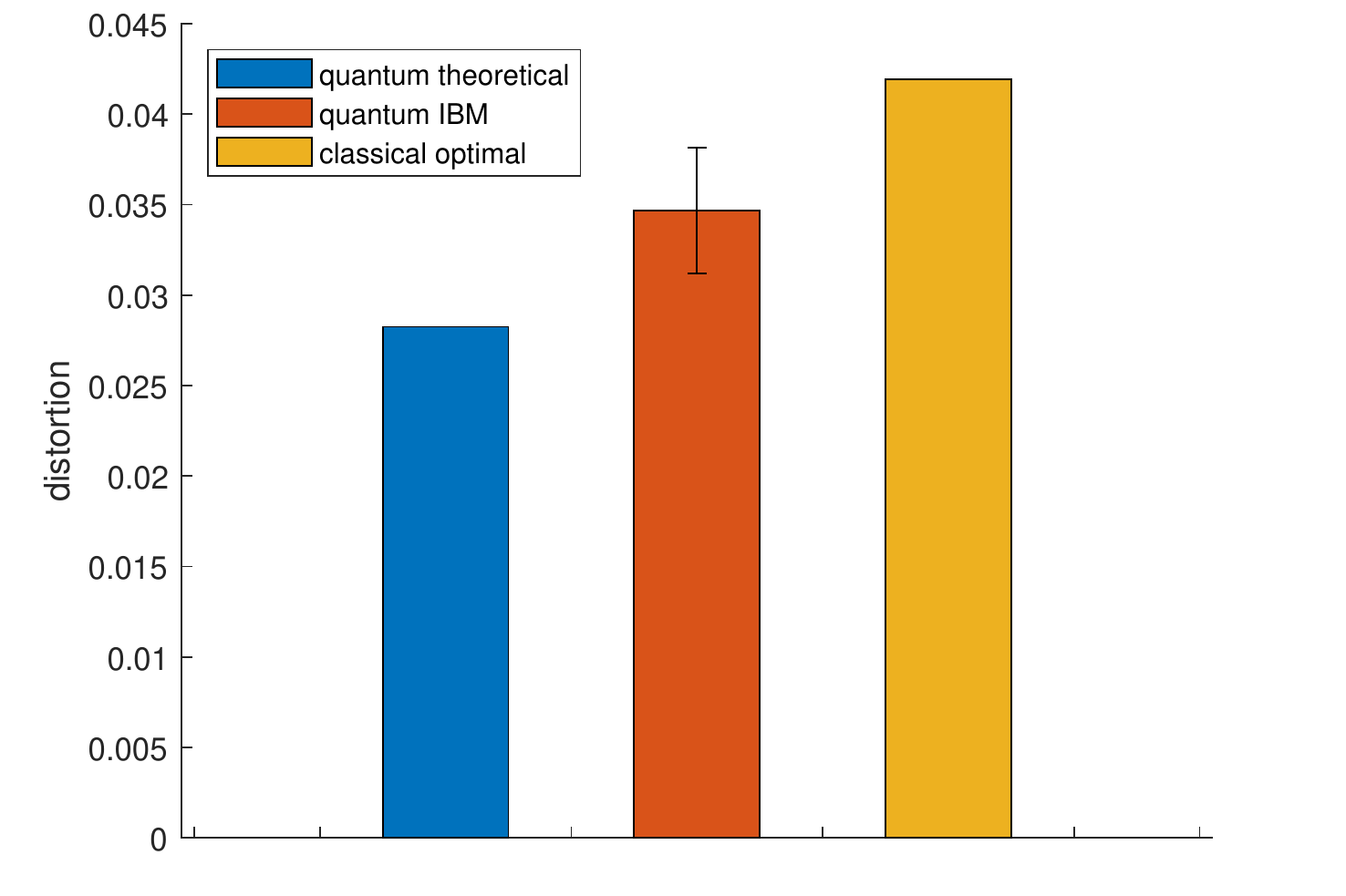}
	\caption{\caphead{Accuracy advantage on present-day quantum hardware.} Despite the experimental noise, the quantum model (orange) retains an accuracy advantage of $2.12$ standard deviations over the classical bound (yellow) for the given dimension.
	For comparison, the blue bar shows the expected distortion of the learned quantum model without experimental noise.
	See \mmsec{} for details.
	}
	\label{fig:IBMdist}
\end{figure}

\begin{figure}[htbp]
	\centering
	\subfloat[][]{\includegraphics[width=0.85\linewidth]{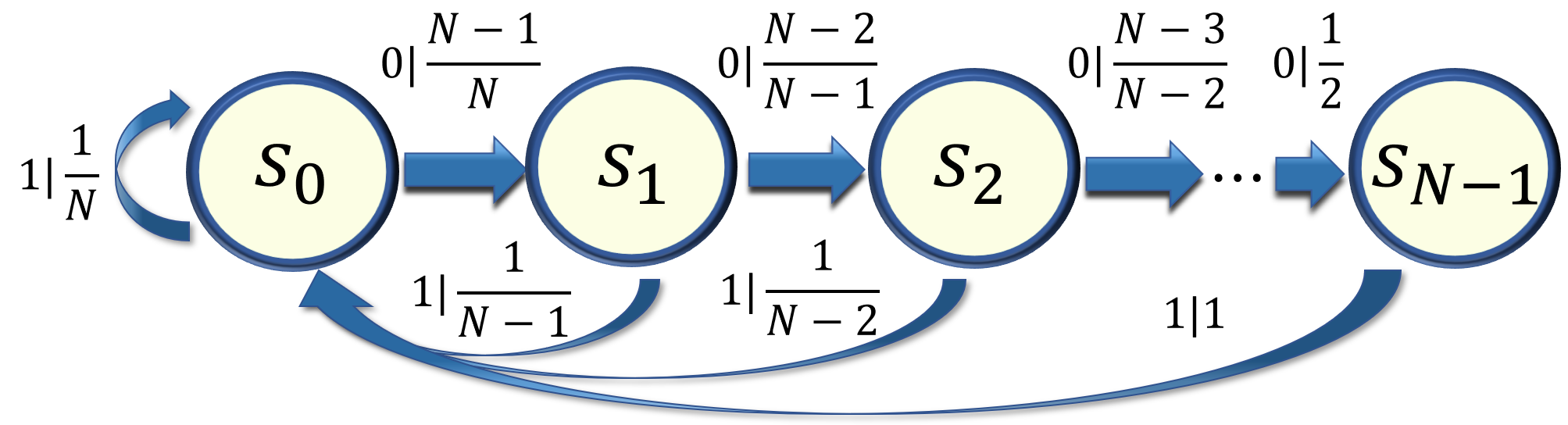}\label{fig:UniformRenewalProcess}}\\
	\subfloat[][]{\includegraphics[width=0.95\linewidth]{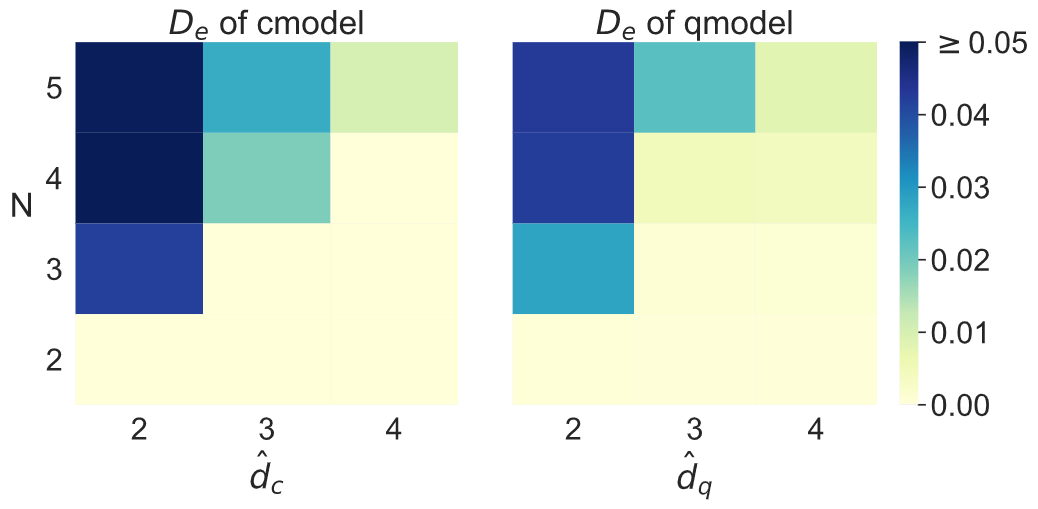}\label{fig:ex3error}}
	\caption{\caphead{Example: The uniform renewal process.} \textbf{(a)} $\varepsilon$-machine for an $N$-step discretization of the uniform renewal process (in which events are spaced by a uniform random interval).
The nodes are the causal states $s_0$, \ldots $s_{N-1}$.
The arrows indicate transitions between states, where each label $x|P(x)$ indicates that output $x$ is produced with probability~$P(x)$.
 \textbf{(b)} Comparison of classical and quantum errors. The left-hand-side shows the lower bound on the optimal approximate classical model's error; the right-hand-side shows the error of the most accurate quantum model found by our algorithm. The accuracy of the resulting quantum model is restricted by finite data whereas the classical models are assumed with full knowledge of the stochastic processes. }
\end{figure}
\section*{Discussion}
Inferring models to predict future behaviour based on past observations is a key task permeating quantitative science. Ideally, when two candidate models require storing the same amount of data from the past, the one that offers predictions with greater statistical fidelity is preferred. 
Here we demonstrate that quantum models have a provable advantage, and present an algorithm to infer such advantaged models directly from data. 
When given training data in terms of a data sequence and a memory of limited dimension, our algorithm is able to encode relevant past information into the memory and discover a quantum process that can sequentially output predictions. 
The accuracy of these predictions -- as measured by their statistical fidelity to the true process generating the training data -- can greatly exceed that of provably optimal classical counterparts. 
The accuracy advantage is realized on an IBM quantum computer, illustrating potential near-term applicability. 
Our work substantially expands upon prior work in memory-enhanced quantum models, which focus on the memory advantages in the exact modelling of a given stochastic process~\cite{Ghafari19a, elliott2019extreme, Thompson2017}. 
The results here illustrate that when the memory dimension is fixed, quantum models can exhibit a significant accuracy advantage. 
This establishes the advantage of quantum models in real-world settings, where there is a trade-off between model complexity and accuracy. 
In this context, our results provide a quantum counterpart to classical tools for structural inference~\cite{shalizi2002algorithm, hinton1999unsupervised, shalizi2014blind} and dimensional reduction in stochastic modelling~\cite{derMaaten09}.

The resulting circuits for executing such models simply involve sequential application of the same unitary process $U$ at each time-step. 
The dimensionality of this unitary is then bounded by the memory dimension of the quantum model. 
As such, a fixed memory dimension guarantees that the simulation of a process for $T$ time-steps has gate complexity that scales linearly with $T$. 
Moreover, unitarity means such quantum models can generate a quantum superposition of all possible futures if the outputs remain unmeasured. 
Thus the quantum model produces a key resource for various quantum data analytics algorithms, such as amplitude estimation, Grover's search, value at risk, and importance sampling~\cite{grover1997quantum,brassard2002quantum,woerner2019quantum,blank2020quantum}. 
The capacity for our models to generate data by sequential application of the same unitary makes them particularly amenable to certain contemporary hardware architectures (e.g.,\ loop-based photonics where only a single optical circuit implementing $U$ needs to be engineered~\cite{schreiber2010photons}). 
Our techniques thus pave the way for quantum-enhanced tools to analyse time-series data from diverse fields such as financial forecasting~\cite{park2007complexity,yang2008increasing}, neuroscience~\cite{Haslinger2009,munoz2020general}, seismology~\cite{Bertello08,Chelidze07}, and natural language processing~\cite{Rabiner89}.

There are a number of promising avenues for future research. 
One caveat of our present learning algorithms is that they remain classical. 
They can thus suffer the curse of dimensionality, and perform best when studying processes that are close to those with low quantum topological complexity. 
While some processes within this class have an arbitrary number of causal states and are already sufficient to demonstrate scaling advantage~\cite{elliott2020extreme}, there is a clear motivation to develop alternative algorithms to learn quantum models of higher memory dimension. 
One possibility is to adopt variational circuit techniques, where candidate quantum models are parameterized by repeated application of some low-depth circuit. 
One can then leverage potential interferometric means to efficiently ascertain distortion in candidate models~\cite{ghafari2019interfering}. 
There are also a number of connections to neighbouring fields.  
One is the dimensional advantages of quantum clocks~\cite{yang2019accuracy,woods2018quantum}, as a clock can be considered a continuum limit of a particular discrete renewal process. 
Meanwhile prediction also plays a key role in information engines that extract free energy from structured patterns~\cite{StillSBC12,Garner2017thermodynamics,Garner19,boyd2018thermodynamics,boyd2020thermodynamic}, where the KL divergence between expectation and reality induces unavoidable energetic waste~\cite{kolchinsky2017dependence,riechers2021initial}. 
Thus our work may herald more effective quantum engines in memory-limited regimes, even when extracting free energy from purely classical correlations.

\section*{\mmsec{}}

\inlineheading{Model Parametrization}
The candidate quantum models of dimension $\hat{d}_q$ for a process with output alphabet $\mathcal{X}$ can be parametrised by a complete set of $|\mathcal{X}|$ Kraus operators $A = \{A_x\}_{x \in \mathcal{X}}$, where $A_x$ denotes the Kraus operator describing how the model updates upon emission of output $x$.
However, this parametrisation is non-ideal for optimisation due to the completeness constraint.

Here, we demonstrate an alternative parametrisation using $B = \{B_x\}_{x \in \mathcal{X}}$, where each $B_x$ is a general $\hat{d}_q \times \hat{d}_q$ complex matrix. Notably, we show that given any such $B$, it is possible to recover a corresponding set $A$ via the following process: First, consider the linear map $\mathcal{E}^{B^*}\!\left(\cdot\right) := \sum_x B_x^\dagger \cdot B_x$. By construction, this map is completely positive, so its leading eigenvalue $\lambda$ is real and positive, and the associated eigenmatrix $V$ is positive semi-definite. Thus, $V$ admits a decomposition $V= W^\dagger W$ (where $W$ is invertible). We can then set $A_x$ such that
\begin{equation}
	A_x := W B_x {W}^{-1}/\sqrt{\lambda}. \label{eq:CompleteKraus}
\end{equation}
Every set $\{A_x\}$ formed in this way satisfies the completeness relation (see \sm~\ref{sm:CompProof}). Thus, each $B_x$ can be used directly to construct a valid set of corresponding Kraus operators $A_x$. This further results in valid unitary operator $U$ (see \sm~\ref{sm:BuildingUnitary}).
Likewise, via \eqref{eq:CompleteKraus} one can also infer the encoding map $\mathcal{E}$ from $B$ (see \sm~\ref{sem:QEM}).

\inlineheading{Computing the Cost Function}
To see how to compute the cost function\\ $C(B) = -\log_2 P_B(x_{0:L})$ directly from $B$, first consider a model with Kraus operators $A = \{A_x\}$ initialized in state $\rho_0$ at $t = 0$. The probability it outputs $x$ at $t = 1$ is then given by $P(x|\rho_0)|_A = \mathrm{Tr}(A_x \rho_0 (A_{x})^\dagger)$, whereby the state transitions to $\rho_1 = A_x \rho_0 (A_{x})^\dag$. As such, repeated iterations of the model will output $x_{0:L}$ with probability
\begin{equation}
	P(x_{0:L}|\rho_0)|_A = \mathrm{Tr}(A_{x_{0:L}}\rho_0A_{x_{0:L}}^\dagger),
	\label{eq:SeqProb}
\end{equation}
where $A_{x_{0:L}} = A_{x_{L-1}}\cdots A_{x_0}$. In addition, if $\rho_0 $ is the stationary memory state averaged over all histories of the process, we obtain the probability of output sequence $x_{0:L}$ when averaged over all pasts. We can then write this likelihood directly in terms of $B^x$ by applying~\eqref{eq:CompleteKraus},
\begin{equation}
	P_B(x_{0:L}) = \mathrm{Tr}(B_{x_{0:L}}\tilde{\rho}{B_{x_{0:L}}}^\dagger V)/\lambda^L,
	\label{Eq: Channel Distribution}
\end{equation}
where $\tilde{\rho}=W^{-1}\bar{\rho}{W^{-1}}^\dagger$ is the leading eigenmatrix of $\mathcal{E}^B$.

\inlineheading{Training process}
Here we specify the details of the training process to minimise the cost function $C(B)= -\log_2 P_B(x_{0:L})$. At this stage, any number of different optimisation techniques could be employed. Specifically, we used Adam optimisation~\cite{kingma2014adam}.

The Adam method is a sophisticated form of gradient descent. Recall that standard gradient descent involves computing $\nabla C$, the partial derivatives with respect to each degree of freedom in $B$ (referred to as a free parameter), and updating $B$ according to $B' = B - \eta \nabla C$, where $\eta > 0 $ is an adjustable learning rate. The Adam method fine-tunes this by using individual adaptive learning rates for each free parameter, computed using estimates of the first and second moments of the gradients $\nabla C$~\cite{kingma2014adam}. Due to the capacity to relate the cost function directly to $B$ (see \eqref{Eq: Channel Distribution}), this can be done without ever needing to distill the quantum circuit itself.

In each run of the training process, we begin with a random $B$. The update process is then repeated until the decrease of updated cost function is below a chosen threshold. Like nearly all gradient descent variants, Adam can converge to local minima. However, in our experiments we found roughly three quarters of all runs give the same minimal value for $C$, and no further gains were made after taking the optimal of three runs (see \sm~\ref{sm:hyperparam} and \ref{sm:learned}).
\begin{figure}[htbp]
\vspace*{0.5em}
\begin{tcolorbox}[skin=enhanced]
 {\bf Box 2: Finding a lower bound on distortion\\of unifilar models}

 \vspace*{0.5em}
 {\bf Inputs:}\\
 \vspace*{0.25em}
 \begin{tabular}{rcl}
	$P$  &  --  & the process $P$ with Markov order $\mr$ \\
	$\hat{d}_c$ & -- & the dimension of the classical model.
 \end{tabular}

 \vspace*{0.5em}
 {\bf Outputs:}\\
 \vspace*{0.25em}
 \begin{tabular}{rcl}
	A lower bound on distortion $D_e$ of all predictive models.
 \end{tabular}	

 \vspace*{0.5em}
 {\bf Algorithm:}
	\vspace*{-0.5em}
	\begin{enumerate}[ label = (\alph*)]
		\item Pick one of the encoding $\encod$ merging causal states of process $P$.
		\item Optimize the future statistics $\hat{P}(X_{0:\mr}|\encod(\past{x}))$ to minimize
		$$D_{\mr}(P,\hat{P}) = \sum_{\past{x}} P(\past{x}) D_{\mr}(P,\hat{P}|\past{x}).$$
		\item Pick a new encoding $\encod$ that merges causal states and repeat (a)-(c) until having gone through all encoding that merge causal states.
		\item Return the minimal $D_{\mr}(P,\hat{P})$ among all encodings.
	\end{enumerate}
\end{tcolorbox}
\end{figure}
The result of this algorithm for the examples discussed in this article is given in~\sm~\ref{sm:boundseg}.

\inlineheading{Lower bound on distortion of classical models}
Here we outline the details of evaluating a lower bound on distortion of classical unifilar models with finite dimension $\hat{d}_c$.
We prove that a bound on the minimum distortion of approximate unifilar models can be obtained by merging causal states, and then optimizing over the statistics of the next $\mr$ outputs for a process with Markov order $\mr$ (see \sm~\ref{sm:bounds}), as summarized in Box 2.
These analytical results reduce the search for bounds on the optimal classical model from an infinite--dimensional problem to a finite--dimensional search, which in general can be exhaustively completed numerically.
Furthermore, the bound becomes tight when $\mr=1$ (Theorem~\ref{thm:MarkovTight} of \sm~\ref{sm:bounds}).

\inlineheading{Implementation of quantum models for the discrete renewal process on IBM’s devices}
To investigate the real-world performance of the models found by the algorithm,
we implement one of the discovered quantum models for the discrete renewal process on IBM's superconducting cloud quantum computer. We consider the case of $\hat{d}_q =2, N=3$ and decompose the quantum model into a quantum circuit consisting of single qubit gates and CNOT gates, as shown in~\fref{fig:IBMCircuit}.
The IBM device used was ``ibmq athens". We ran the circuit $40,000$ times for each input quantum state and obtained the output probabilities by measuring the output register (see \sm{}~\ref{app:IBM}).

\begin{figure}[htbp]
	\centering
	\includegraphics[width=0.95\linewidth]{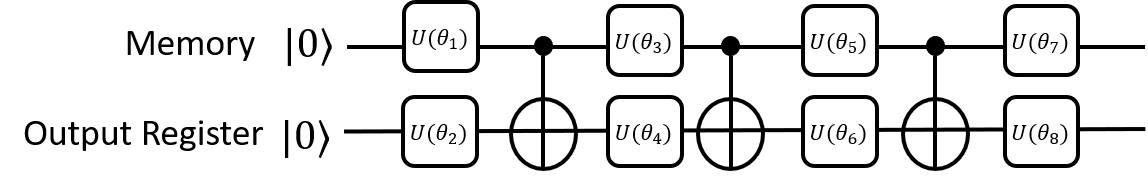}
	\caption{\caphead{Circuit diagram for a single step output.} The upper line represents the memory qubit while the lower line represents the output register. The circuit involves $8$ single qubit gates ($U(\theta_i)$) and $3$ C-NOT gates. Each single qubit gate is specified by three real parameters. $U(\theta_1)$ encodes the input memory state. }
	\label{fig:IBMCircuit}
\end{figure}

\fref{fig:IBMdist} compares the distortion realized on the IBM device with the optimal classical bound for $\hat{d}_c = 2$ and the ideal realization of the learned quantum model.
The 40,000 runs were divided into 50 batches, and the distortion was calculated for each batch. The error bar depicts the standard deviation in distortion of those 50 batches. Our quantum model maintains a statistically significant accuracy advantage even in a noisy environment: there is a $2.12$ standard deviations gap between the experimental quantum distortion and our lower bound on the best classical distortion.

\section*{Acknowledgements}
We are grateful for discussions with Thomas Elliott, Jim Crutchfield and Kelvin Onggadinata. This research is supported by the National Research Foundation (NRF), Singapore, under its NRFF Fellow programme (Award No.\ NRF-NRFF2016-02), the Singapore Ministry of Education Tier 1 Grant No.\ 2018-T1-002-043 (RG162/19) and Grant No.\ FQXI R-710-000-146-720 (Are quantum agents more energetically efficient at making predictions?) from the Foundational Questions Institute and Fetzer Franklin Fund (a donor--advised fund of Silicon Valley Community Foundation), the Quantum Engineering Program QEP-SP3. N.T. acknowledges support by the Griffith University Postdoctoral Fellowship Scheme.


\bibliographystyle{unsrt}
\bibliography{References}

\begin{thebibliography}{10}

\bibitem{van2009dimensionality}
Laurens Van Der~Maaten, Eric Postma, and Jaap Van~den Herik.
\newblock Dimensionality reduction: a comparative.
\newblock {\em J Mach Learn Res}, 10(66-71):13, 2009.

\bibitem{alpaydin2020introduction}
Ethem Alpaydin.
\newblock {\em Introduction to machine learning}.
\newblock MIT press, 2020.

\bibitem{ng2011sparse}
Andrew Ng et~al.
\newblock Sparse autoencoder.
\newblock {\em CS294A Lecture notes}, 72(2011):1--19, 2011.

\bibitem{pearl2000models}
Judea Pearl et~al.
\newblock Models, reasoning and inference.
\newblock {\em Cambridge, UK: CambridgeUniversityPress}, 2000.

\bibitem{Shalizi2001}
Cosma~Rohilla Shalizi and James~P. Crutchfield.
\newblock {Computational mechanics: Pattern and prediction, structure and
  simplicity}.
\newblock {\em Journal of Statistical Physics}, 104(3-4):817--879, 2001.

\bibitem{shalizi2002algorithm}
Cosma~Rohilla Shalizi, Kristina~Lisa Shalizi, and James~P Crutchfield.
\newblock An algorithm for pattern discovery in time series.
\newblock {\em arXiv preprint cs/0210025}, 2002.

\bibitem{shalizi2014blind}
Cosma Shalizi and Kristina~Lisa Klinkner.
\newblock Blind construction of optimal nonlinear recursive predictors for
  discrete sequences.
\newblock {\em arXiv preprint arXiv:1408.2025}, 2014.

\bibitem{Haslinger2009}
Robert Haslinger, Kristina~Lisa Klinkner, and Cosma~Rohilla Shalizi.
\newblock {The Computational Structure of Spike Trains}.
\newblock {\em Neural Computation}, 22(1):121--157, jan 2010.

\bibitem{Clarke2003}
R.~W. Clarke, M.~P. Freeman, and N.~W. Watkins.
\newblock {Application of computational mechanics to the analysis of natural
  data: An example in geomagnetism}.
\newblock {\em Physical Review E}, 67(1):016203, 2003.

\bibitem{munoz2020general}
Roberto~N Mu{\~n}oz, Angus Leung, Aidan Zecevik, Felix~A Pollock, Dror Cohen,
  Bruno van Swinderen, Naotsugu Tsuchiya, and Kavan Modi.
\newblock General anesthesia reduces complexity and temporal asymmetry of the
  informational structures derived from neural recordings in drosophila.
\newblock {\em Physical Review Research}, 2(2):023219, 2020.

\bibitem{Elliott2018Continous}
Thomas~J. Elliott and Mile Gu.
\newblock {Superior memory efficiency of quantum devices for the simulation of
  continuous-time stochastic processes}.
\newblock {\em npj Quantum Information}, 4(1):18, 2018.

\bibitem{elliott2019extreme}
T.~J. Elliott, C.~Yang, F.~C. Binder, A.~J.~P. Garner, J.~Thompson, and M.~Gu.
\newblock {Extreme Dimensionality Reduction with Quantum Modeling}.
\newblock {\em Physical Review Letters}, 125(26):260501, dec 2020.

\bibitem{gu2012quantum}
Mile Gu, Karoline Wiesner, Elisabeth Rieper, and Vlatko Vedral.
\newblock {Quantum mechanics can reduce the complexity of classical models}.
\newblock {\em Nature Communications}, 3:762, 2012.

\bibitem{Aghamohammadi2017PRX}
C.~Aghamohammadi, S.~P. Loomis, J.~R. Mahoney, and J.~P. Crutchfield.
\newblock {Extreme Quantum Advantage for Rare-Event Sampling}.
\newblock {\em Physical Review X}, 8(1):11025, 2017.

\bibitem{Binder2017}
Felix~C Binder, Jayne Thompson, and Mile Gu.
\newblock {Practical Unitary Simulator for Non-Markovian Complex Processes}.
\newblock {\em Physical Review Letters}, 120(24):240502, 2018.

\bibitem{yang2018matrix}
Chengran Yang, Felix~C. Binder, Varun Narasimhachar, and Mile Gu.
\newblock Matrix product states for quantum stochastic modeling.
\newblock {\em Phys. Rev. Lett.}, 121:260602, Dec 2018.

\bibitem{Ghafari19a}
Farzad Ghafari, Nora Tischler, Jayne Thompson, Mile Gu, Lynden~K. Shalm,
  Varun~B. Verma, Sae~Woo Nam, Raj~B. Patel, Howard~M. Wiseman, and Geoff~J.
  Pryde.
\newblock Dimensional quantum memory advantage in the simulation of stochastic
  processes.
\newblock {\em Phys. Rev. X}, 9:041013, Oct 2019.

\bibitem{liu2019optimal}
Qing Liu, Thomas~J Elliott, Felix~C Binder, Carlo Di~Franco, and Mile Gu.
\newblock Optimal stochastic modeling with unitary quantum dynamics.
\newblock {\em Physical Review A}, 99(6):062110, 2019.

\bibitem{Thompson2017}
Jayne Thompson, Andrew J.~P. Garner, John~R. Mahoney, James~P. Crutchfield,
  Vlatko Vedral, and Mile Gu.
\newblock {Causal Asymmetry in a Quantum World}.
\newblock {\em Physical Review X}, 8(3):31013, 2018.

\bibitem{Crutchfield1989}
James~P. Crutchfield and Karl Young.
\newblock {Inferring statistical complexity}.
\newblock {\em Physical Review Letters}, 63(2):105--108, 1989.

\bibitem{yang2008increasing}
Jae~Suk Yang, Wooseop Kwak, Taisei Kaizoji, and In~Mook Kim.
\newblock {Increasing market efficiency in the stock markets}.
\newblock {\em European Physical Journal B}, 61(2):241--246, 2008.

\bibitem{park2007complexity}
Joongwoo~Brian Park, Jeong {Won Lee}, Jae~Suk Yang, Hang~Hyun Jo, and Hie~Tae
  Moon.
\newblock {Complexity analysis of the stock market}.
\newblock {\em Physica A: Statistical Mechanics and its Applications},
  379(1):179--187, 2007.

\bibitem{Boyd2016identifying}
Alexander~B. Boyd, Dibyendu Mandal, and James~P. Crutchfield.
\newblock {Identifying Functional Thermodynamics in Autonomous Maxwellian
  Ratchets}.
\newblock {\em New Journal of Physics}, 18(2):023049, jul 2015.

\bibitem{Garner2017thermodynamics}
Andrew~JP Garner, Jayne Thompson, Vlatko Vedral, and Mile Gu.
\newblock Thermodynamics of complexity and pattern manipulation.
\newblock {\em Physical Review E}, 95(4):042140, 2017.

\bibitem{marzen2015informational}
Sarah~E Marzen and James~P Crutchfield.
\newblock Informational and causal architecture of discrete-time renewal
  processes.
\newblock {\em Entropy}, 17(7):4891--4917, 2015.

\bibitem{hobson1987concepts}
Art Hobson.
\newblock {\em Concepts in statistical mechanics}.
\newblock CRC Press, 1987.

\bibitem{cover2012elements}
Thomas~M Cover and Joy~A Thomas.
\newblock {\em Elements of information theory}.
\newblock John Wiley \& Sons, 2012.

\bibitem{partghafari2019dimensional}
Farzad Ghafari, Nora Tischler, Jayne Thompson, Mile Gu, Lynden~K. Shalm,
  Varun~B. Verma, Sae~Woo Nam, Raj~B. Patel, Howard~M. Wiseman, and Geoff~J.
  Pryde.
\newblock Dimensional quantum memory advantage in the simulation of stochastic
  processes.
\newblock {\em Phys. Rev. X}, 9:041013, Oct 2019.

\bibitem{loomis2019strong}
Samuel~P Loomis and James~P Crutchfield.
\newblock Strong and weak optimizations in classical and quantum models of
  stochastic processes.
\newblock {\em Journal of Statistical Physics}, 176(6):1317--1342, 2019.

\bibitem{tegmark2020pareto}
Max Tegmark and Tailin Wu.
\newblock Pareto-optimal data compression for binary classification tasks.
\newblock {\em Entropy}, 22(1):7, 2020.

\bibitem{marzen2016predictive}
Sarah~E Marzen and James~P Crutchfield.
\newblock Predictive rate-distortion for infinite-order markov processes.
\newblock {\em Journal of Statistical Physics}, 163(6):1312--1338, 2016.

\bibitem{marzen2017nearly}
Sarah~E Marzen and James~P Crutchfield.
\newblock Nearly maximally predictive features and their dimensions.
\newblock {\em Physical Review E}, 95(5):051301, 2017.

\bibitem{Vidal2007}
G.~Vidal.
\newblock {Entanglement renormalization}.
\newblock {\em Physical Review Letters}, 99(22):220405, 2007.

\bibitem{petrie1969probabilistic}
Ted Petrie.
\newblock Probabilistic functions of finite state markov chains.
\newblock {\em The Annals of Mathematical Statistics}, 40(1):97--115, 1969.

\bibitem{Juang1985}
B.~Juang and L.~R. Rabiner.
\newblock A probabilistic distance measure for hidden markov models.
\newblock {\em AT T Technical Journal}, 64(2):391--408, Feb 1985.

\bibitem{kelly2012new}
David Kelly, Mark Dillingham, Andrew Hudson, and Karoline Wiesner.
\newblock A new method for inferring hidden markov models from noisy time
  sequences.
\newblock {\em PloS one}, 7(1):e29703, 2012.

\bibitem{HorodeckiOppenheim2013}
M.~Horodecki and J.~Oppenheim.
\newblock {Fundamental limitations for quantum and nanoscale thermodynamics}.
\newblock {\em Nat Commun}, (4):2059, 2013.

\bibitem{doob1948renewal}
Joseph~L Doob.
\newblock Renewal theory from the point of view of the theory of probability.
\newblock {\em Transactions of the American Mathematical Society},
  63(3):422--438, 1948.

\bibitem{crutchfield2015time}
James~P Crutchfield, Michael~Robert DeWeese, and Sarah~E Marzen.
\newblock Time resolution dependence of information measures for spiking
  neurons: Scaling and universality.
\newblock {\em Frontiers in computational neuroscience}, 9:105, 2015.

\bibitem{Marzen2017}
Sarah Marzen and James~P Crutchfield.
\newblock {Informational and Causal Architecture of Continuous-time Renewal
  Processes}.
\newblock {\em Journal of Statistical Physics}, 168(1):109--127, 2017.

\bibitem{hinton1999unsupervised}
Geoffrey~E Hinton, Terrence~Joseph Sejnowski, Tomaso~A Poggio, et~al.
\newblock {\em Unsupervised learning: foundations of neural computation}.
\newblock MIT press, 1999.

\bibitem{derMaaten09}
L.~van~der Maaten, E.~Postma, and J.~van~den Herik.
\newblock {Dimensionality Reduction: A Comparative Review}.
\newblock {\em Technical Report, Tilburg University}, 2007.

\bibitem{grover1997quantum}
Lov~K Grover.
\newblock Quantum mechanics helps in searching for a needle in a haystack.
\newblock {\em Physical review letters}, 79(2):325, 1997.

\bibitem{brassard2002quantum}
Gilles Brassard, Peter Hoyer, Michele Mosca, and Alain Tapp.
\newblock Quantum amplitude amplification and estimation.
\newblock {\em Contemporary Mathematics}, 305:53--74, 2002.

\bibitem{woerner2019quantum}
Stefan Woerner and Daniel~J Egger.
\newblock Quantum risk analysis.
\newblock {\em npj Quantum Information}, 5(1):1--8, 2019.

\bibitem{blank2020quantum}
Carsten Blank, Daniel~K Park, and Francesco Petruccione.
\newblock Quantum-enhanced analysis of discrete stochastic processes.
\newblock {\em arXiv preprint arXiv:2008.06443}, 2020.

\bibitem{schreiber2010photons}
A.~Schreiber, K.~N. Cassemiro, V.~Poto\ifmmode~\check{c}\else \v{c}\fi{}ek,
  A.~G\'abris, P.~J. Mosley, E.~Andersson, I.~Jex, and Ch. Silberhorn.
\newblock Photons walking the line: A quantum walk with adjustable coin
  operations.
\newblock {\em Phys. Rev. Lett.}, 104:050502, Feb 2010.

\bibitem{Bertello08}
G.~Bertello, P.~J. Arduin, F.~Boschetti, and D.~Weatherley.
\newblock {Application of Computational Mechanics to the analysis of seismic
  time-series via numerical optimisation}.
\newblock {\em New Generation Computing}, 27(1):1--23, nov 2008.

\bibitem{Chelidze07}
T.~Chelidze and T.~Matcharashvili.
\newblock {Complexity of seismic process; measuring and applications - A
  review}.
\newblock {\em Tectonophysics}, 431(1-4):49--60, feb 2007.

\bibitem{Rabiner89}
Lawrence~R. Rabiner.
\newblock {A Tutorial on Hidden Markov Models and Selected Applications in
  Speech Recognition}.
\newblock {\em Proceedings of the IEEE}, 77(2):257--286, 1989.

\bibitem{elliott2020extreme}
Thomas~J Elliott, Chengran Yang, Felix~C Binder, Andrew~JP Garner, Jayne
  Thompson, and Mile Gu.
\newblock Extreme dimensionality reduction with quantum modeling.
\newblock {\em Physical Review Letters}, 125(26):260501, 2020.

\bibitem{ghafari2019interfering}
Farzad Ghafari, Nora Tischler, Carlo Di~Franco, Jayne Thompson, Mile Gu, and
  Geoff~J Pryde.
\newblock Interfering trajectories in experimental quantum-enhanced stochastic
  simulation.
\newblock {\em Nature communications}, 10(1):1--8, 2019.

\bibitem{yang2019accuracy}
Yuxiang Yang, Lennart Baumg{\"a}rtner, Ralph Silva, and Renato Renner.
\newblock Accuracy enhancing protocols for quantum clocks.
\newblock {\em arXiv preprint arXiv:1905.09707}, 2019.

\bibitem{woods2018quantum}
Mischa~P Woods, Ralph Silva, Gilles P{\"u}tz, Sandra Stupar, and Renato Renner.
\newblock Quantum clocks are more accurate than classical ones.
\newblock {\em arXiv preprint arXiv:1806.00491}, 2018.

\bibitem{StillSBC12}
S.~Still, D.~A. Sivak, A.~J. Bell, and G.~E. Crooks.
\newblock {Thermodynamics of Prediction}.
\newblock {\em Physical Review Letters}, 109(12):120604, sep 2012.

\bibitem{Garner19}
A.~J.~P. Garner.
\newblock {The fundamental thermodynamic bounds on finite models}.
\newblock {\em arXiv preprint arXiv:1912.03217}, dec 2019.

\bibitem{boyd2018thermodynamics}
Alexander~B. Boyd, Dibyendu Mandal, and James~P. Crutchfield.
\newblock Thermodynamics of modularity: Structural costs beyond the landauer
  bound.
\newblock {\em Phys. Rev. X}, 8:031036, Aug 2018.

\bibitem{boyd2020thermodynamic}
Alexander~B Boyd, James~P Crutchfield, and Mile Gu.
\newblock Thermodynamic machine learning through maximum work production.
\newblock {\em arXiv preprint arXiv:2006.15416}, 2020.

\bibitem{kolchinsky2017dependence}
Artemy Kolchinsky and David~H Wolpert.
\newblock Dependence of dissipation on the initial distribution over states.
\newblock {\em Journal of Statistical Mechanics: Theory and Experiment},
  2017(8):083202, 2017.

\bibitem{riechers2021initial}
Paul~M Riechers and Mile Gu.
\newblock Initial-state dependence of thermodynamic dissipation for any quantum
  process.
\newblock {\em Physical Review E}, 103(4):042145, 2021.

\bibitem{kingma2014adam}
Diederik~P Kingma and Jimmy Ba.
\newblock Adam: A method for stochastic optimization.
\newblock {\em arXiv preprint arXiv:1412.6980}, 2014.

\bibitem{horn2012matrix}
Roger~A Horn and Charles~R Johnson.
\newblock {\em Matrix analysis}.
\newblock Cambridge university press, 2012.

\bibitem{sarah2017nearly}
Sarah~E. Marzen and James~P. Crutchfield.
\newblock Nearly maximally predictive features and their dimensions.
\newblock {\em Phys. Rev. E}, 95:051301, May 2017.

\end{thebibliography}
\balancecolsandclearpage


\appendix

\section{Proof of Completeness}
\label{sm:CompProof}
Here we prove that for each model parametrised by $B = \{B_x\}_x$, the set $A = \{A_x\}_x$ is complete. This follows from
\begin{align}
\sum_x A_x^\dagger A_x & = \sum_x \frac{1}{\lambda}  ({W}^{-1})^{\dagger}{B_x}^\dagger W^\dag W B_x W^{-1} \nonumber \\
 & = \frac{1}{\lambda}  ({W}^\dagger)^{-1} \left( \sum_x {B_x}^\dagger W^\dagger W B_x \right) W^{-1} \nonumber \\
 & = ({W}^{\dagger})^{-1} W^\dagger W W^{-1} = \mathbb{I},
\end{align}
where we use $({W}^{\dagger})^{-1} = ({W}^{-1})^{\dagger}$.

\section{Building Unitary Circuits}
\label{sm:BuildingUnitary}
Once we have the Kraus operators $A_x$ that describe a model, standard techniques enable construction of a full unitary circuit. Specifically, we need to find a $|\mathcal{X}|\hat{d}_q \times |\mathcal{X}|\hat{d}_q$ unitary operator $U$ such that
\begin{equation}
	\bra{x}U\ket{0} = A_x.
\end{equation}
As a result, some of the elements of the unitary operators are predefined, i.e.,
\begin{equation}
	\bra{j}\bra{x}U\ket{k}\ket{0} = A_{x,jk}.
\end{equation}
These predefined elements forms $\hat{d}_q$ columns of the unitary operators, i.e., $\ket{\phi_{k0}}:= U\ket{k}\ket{0}$.
Such columns $\ket{\phi_{k0}}$ are mutually orthogonal quantum states, due to the completeness relation of $A_x$. Since an operator $U$ is unitary if and only if its columns form an orthogonal basis, the remaining task is to find the complementary quantum states $\ket{\phi_{kx}}$ so that they form orthogonal states.
This can be done by the Gram-Schmidt process~\cite{horn2012matrix}.

\section{Quantum encoding maps}
\label{sem:QEM}
The encoding map $\encod$ should map the past $\past{x}$ onto a pure quantum state $\ket{\sigma_{\past{x}}}$, which is used as the initial configuration of the quantum model.
Assume first that the process has finite Markov order $\mr$ such that only the latest $\mr$ steps of a given past $\past{x}$ affect future statistics, i.e., $P(\future{X}|\past{x}) \equiv P(\future{X}|x_{-\mr:0})$.
Writing $0$ as the element of $\mathcal{X}$ that also coincides with the initial state of the output register, for the purpose of the model, any given past $\past{x}$ is equivalent to the concatenation $0^{\infty}, x_{-\mr:0}$, where $0^{\infty} := \cdots 0,0,0$.
That is,
\begin{equation}
	P(\future{x}|\past{x})=P(\future{x}| 0^{\infty}, x_{-\mr:0} ).
\end{equation}
Recall that an output of $0^{\infty}$ corresponds to applying Kraus operator $A_0$ infinitely many times on the initial quantum state.
As a result, any initial quantum state converges to its leading eigenvector $\ket{\sigma_0}$ of $A_0$ after applying $A_0$ an infinite number of times.
The quantum state transitions to $\ket{\sigma_{x_{-\mr:0}}} :=A_{x_{-\mr:0}}\ket{\sigma_0}$ after outputting $x_{-\mr:0}$.
The choice of output symbol here was arbitrary, as such, the argument can be repeated for any output $x$ in place of $0$.

When the process has infinite Markov order, perfect initialization of the machine would require an observation of the entire process history. However, one can approximate by taking a long but finite string, which will have exponentially diminishing error with the length of the string~\cite{sarah2017nearly}. As such, the above choice of $\ket{\sigma_0}$ remains valid.

\section{Hyperparameters}
\label{sm:hyperparam}
The length of the sample in the first two examples was $1000$, while that of the third example was $5000$.
The third example -- the renewal process -- has more causal states, and training quantum models with higher $\hat{d}_q$ also requires more data.
For all of the three examples, the learning rate $\eta$ is $0.1$ and the optimizer is Adam~\cite{kingma2014adam}.
We repeat the training three times and choose the quantum model of minimal cost.

The evaluation of the error $D_e(P,\hat{P})$ relies on computing the KL-divergence \\$\mathcal{D}_{\rm KL}(P(X_{0:L}|\past{x})||\hat{P}(X_{0:L}|\epsilon(\past{x}))$.
Here, the length of the past $\past{x}$ is $5$ and the length of the future $L$ is $1$ (see~\sm~G).
This amounts to assuming that the Markov order of the process is less than or equal to 5, which is true by construction for our examples. The update process is then repeated until $\Delta C$ is less than $0.1$.

\section{Learned Quantum Models}
\label{sm:learned}
Here we present details of the quantum models produced by our model inference algorithm. We show the Kraus operators together with the corresponding encoding quantum states. As $\hat{d}_q$ is set to $2$, each encoding quantum state is represented by a point in a Bloch sphere.

\inlinesubheading{Case 1: the asymmetric process}
For a $\hat{d}_q=2$ model of the asymmetric process (recall \fref{Fig:asymm}) with $p=0.3, q = 0.8$, we found the Kraus operators:
\begin{equation}
	\begin{split}
	&A_0 = \begin{bmatrix}
		0.676 &  0.317\\
		0.316 &  0.150
	\end{bmatrix},\\
	&A_1 = \begin{bmatrix}
		-0.264 &  0.534\\
		-0.336 &  0.729
	\end{bmatrix},\\
	&A_2 = \begin{bmatrix}
		-0.241 & -0.095\\
		0.449 &  0.227
	\end{bmatrix}.\\
	\end{split}
\end{equation}
For the asymmetric processes, real amplitudes are sufficient to train a quantum model with $\hat{d}_q = 2$ to good accuracy~\cite{Thompson2017}.

The quantum states $\ket{\sigma_{\past{x}}}$ associated with encodings of pasts of length $5$ are plotted in~\fref{fig:ex1states}.
The more likely quantum states lie in three clusters.
This coincides with our expectation from the classical $\varepsilon$-machine, which has $3$ causal states-- but here, this has been discovered directly from the sample during training. Moreover, the location of these states aligns with our expectation from the theoretical optimal quantum states ascertained in \cite{Thompson2017}:
 two clusters (pink and blue) approximately correspond to orthogonal quantum states, while the third (green) lies biased between the two.

\begin{figure}[htb]
	\centering
	\includegraphics[width=0.95\linewidth]{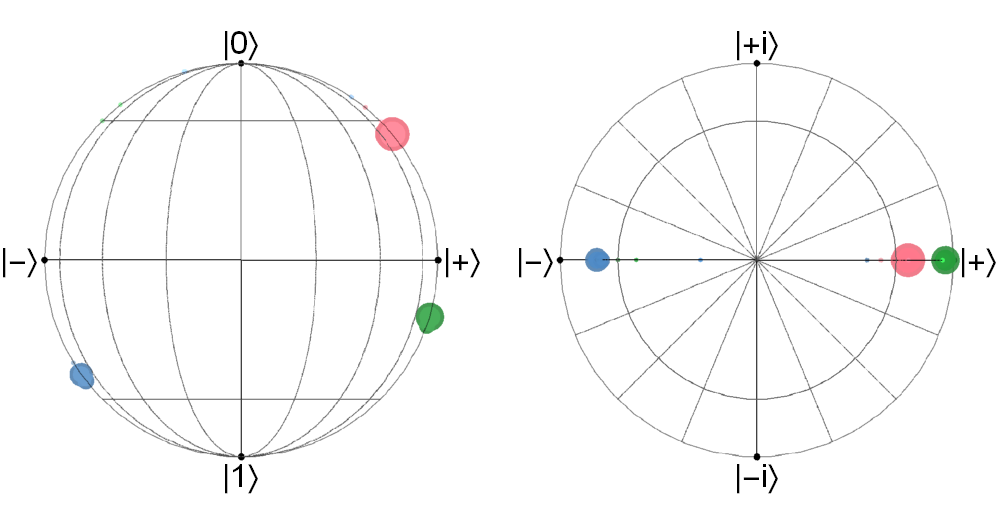}
	\caption{
\caphead{Example: Asymmetric process -- encoded states.}
	Quantum states produced by the learned $\hat{d}_q=2$ encoding map for the asymmetric process with $p=0.3$, $q=0.8$,
	represented on a Bloch sphere (viewed sideways and top-down).
Each point represents a pure state mapped to by one history (of length $5$),
 and the colour represents the history's associated causal state.
For larger points (representing the most probable $99\%$ of histories), the area is proportional to the probability of that history.
For smaller points (representing the next $0.999\%$ of histories), the opacity is proportional to the probability.}
	\label{fig:ex1states}
\end{figure}

\inlinesubheading{Case 2: the quasi-cycle}
For the quasi-cycle, we show the case where $p=0.5,\delta =0.1$.
The Kraus operators found for $\hat{d}_q=2$ are
\begin{equation}
	\begin{split}
		&A_0 = \begin{bmatrix}
		0.021-0.077i & -0.104+0.084i\\
		-0.277+0.236i &  0.656+0.098i\\
		\end{bmatrix},\\
		&A_1 = \begin{bmatrix}
		0.508+0.135i& -0.233+0.465i\\
		0.279-0.171i&  0.073+0.277i\\
		\end{bmatrix},\\
		&A_2 = \begin{bmatrix}
		0.371+0.271i&  0.304+0.010i\\
		-0.209+0.473i&  0.055+0.308i\\
		\end{bmatrix}.
	\end{split}
\end{equation}
The quantum states $\ket{\sigma_{\past{x}}}$ are plotted in~\fref{fig:ex2states}.
As shown in~\fref{fig:ex2states}, most of the states (weighted by probability of occurrence) lie in three distinct clusters -- but here, no pair of clusters is orthogonal.

\begin{figure}[htbp]
	\centering
	\includegraphics[width=0.95\linewidth]{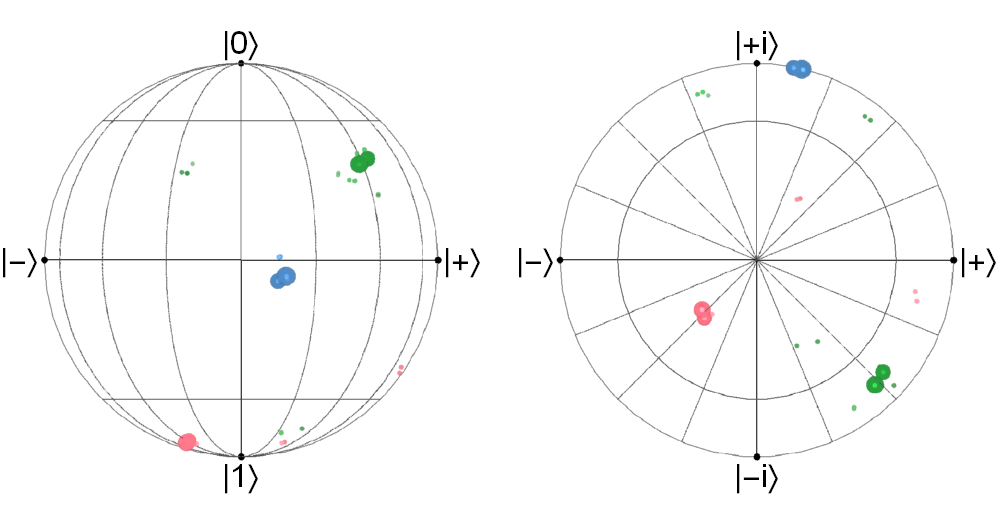}
		\caption{\caphead{Example: Quasicycle -- encoded states.}
		Quantum states produced by the learned $\hat{d}_q=2$ encoding map for a quasicycle with $p=0.5,\delta =0.1$, represented on a Bloch sphere. The diagram is interpreted as in \fref{fig:ex1states}.
	}
	\label{fig:ex2states}
\end{figure}

\inlinesubheading{Case 3: the uniform renewal process}
For the discrete renewal process (\fref{fig:UniformRenewalProcess}) we show the results for $N=3, \hat{d}_q= 2$.
The Kraus operators found are
\begin{equation}
	\begin{split}
		&A_0 = \begin{bmatrix}
			0.064+0.170i & 0.043 + 0.246i\\
			-0.196 + 0.825i & 0.499 - 0.079i
		\end{bmatrix},\\
		&A_1 = \begin{bmatrix}
			0.490 - 0.053i & 0.304 + 0.753i \\
			0.005 - 0.068i & 0.048 - 0.142i
		\end{bmatrix}
	\end{split}
\end{equation}
The associated encoded states are plotted in \fref{fig:ex3states}.

\begin{figure}[htbp]
	\centering
	\includegraphics[width=0.95\linewidth]{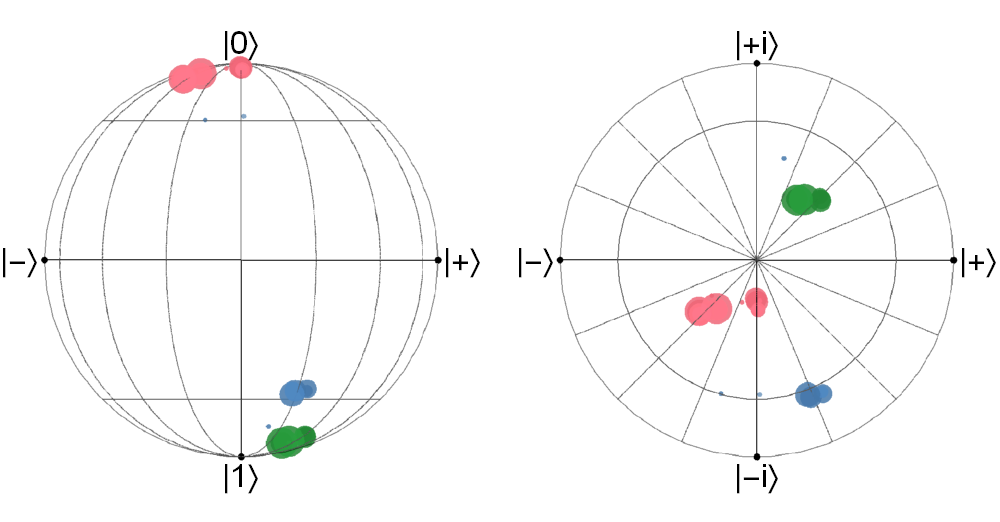}
	\caption{\caphead{Example: Renewal process -- encoded states.}
	Quantum states produced by the learned $\hat{d}_q=2$ encoding map for an $N=3$ discretization of the uniform renewal process, represented on a Bloch sphere. The diagram is interpreted as in \fref{fig:ex1states}.
	}
	\label{fig:ex3states}
\end{figure}

\section{Bounding the minimum-distortion classical model}
\label{sm:bounds}
\subsection{Past partitions, future morphs, and pre-models}
Let $\past{X}$ be the set of all possible observable histories of some process $P$,
 and $\mathcal{R}$ be a countable set of possible model memory states.
A {\em partition of histories} or {\em encoding} $\encod$ is then a function $\encod:\past{X}\to\mathcal{R}$.
We refer to the number of distinct values in $\mathcal{R}$ taken by $\mathcal{E}$ (excluding those that occur with measure zero with respect to $P(\past{X})$) as the model {\em dimension} of $\encod$, which we write as $D_c$.
That is, each $\encod$ effectively divides the set of histories into $D_c$ mutually exclusive partitions.
To use this encoding as a model, we must also supply a second map $\morph: \mathcal{R}\to\future{X}$ from memory states to a set of statistics over the future of the process known as the {\em future morphs} of the memory states.
Together, the composition of $\encod$ and $\morph$ defines a conditional probability distribution $\hat{P}(\future{X}|\mathcal{E}(\past{x}))$,
 which can be directly compared with a process's conditional probability distribution $P(\future{X}|\past{x})$ for each history $\past{x}$.

As a clarifying example, let us cast the $\varepsilon$--machine into this language.
Here $\mathcal{E} = \varepsilon$, the partition of histories according to the equivalence relation $\past{x} \sim \past{x}'$ if and only if  $P(\future{X}|\past{x}) = P(\future{X}|\past{x}^{'})$.
We write $\mathcal{R} \equiv \mathcal{S}$, where the elements of $\mathcal{S}$ are known as the process's {\em causal states}.
Meanwhile, the future morph associated with any casual state is the same as any constituent pasts within the class, i.e.\ for all $s_i\in\mathcal{S}$, $\hat{P}(\future{X}|s_i) = P(\future{X}|\past{x})$ where $\past{x}$ is an arbitrary choice of the past such that $\varepsilon(\past{x}) = s_i$.

In general, a couple $(\mathcal{E}, \mathcal{F})$ does not always correspond to a stationary hidden Markov model, since there might not be a systematic action $\mathcal{P}$ that can be repeatedly applied to $\mathcal{R}$ to sequentially produce a sample from the morph.
As such, we call the couple $(\encod, \morph)$ a {\em pre-model}:
\begin{definition}[Pre-model]
A {\bf pre-model} is the couple $(\encod, \morph)$ where
 $\encod$ is a function from histories ($\past{X}$) to memory states ($\mathcal{R}$),
 and $\morph$ is a function from memory states ($\mathcal{R}$) to distributions over future statistics ($P(\future{X}|r), \forall r \in \mathcal{R}$).
\end{definition}

We write the set of pre-models with maximum dimension $d$ as $\mathcal{M}^d_\infty$.
This set will be used to derive bounds on the properties of unifilar hidden Markov models, as we shall subsequently show.
First, we define a particular class of pre-model:
\begin{definition}[$K$-unifilar pre-model]
For $K\in\ints^+$, a {\bf $K$-unifilar pre-model} is the triple $(\encod,\morph, \mathcal{P}_K)$ consisting of a pre-model $(\encod, \morph)$ with a deterministic map $\mathcal{P}_K: \mathcal{R}\times \mathcal{X}^{\otimes K} \to\mathcal{R}$ such that $\mathcal{P}_K \! \left(\encod(\past{x}), x_{0:K} \right) = \encod(\past{x} x_{0:K})$ for all $x_{0:K}$ and all $\past{x}$.
\end{definition}
We write the set of all $K$-unifilar pre-models of maximum dimension $d$ as $\mathcal{M}^d_K$.
Crucially, the valid morphs of $K$-unifilar pre-models are highly constrained:
\begin{lemma}
\label{lem:KUnifWordK}
A $K$-unifilar pre-model $(\encod, \morph, \mathcal{P}_K)$ can always be specified by the triple $(\encod,\morph_K, \mathcal{P}_K)$ where $\morph_K$ is a map from $\mathcal{R}$ to probability distributions over words of length $K$.
\begin{proof}
First, let $r:= \mathcal{E}(\past{x})$ and
 note that the definitions of $\mathcal{P}_K$ and $\morph$ imply that
 $\mathcal{F}\left( \mathcal{E}\!\left(\past{x} x_{0:K} \right) \right) =
 \mathcal{F}\left( \mathcal{P}_K\!\left(\mathcal{E}\!\left(\past{x}\right), x_{0:K} \right) \right)$ allowing
 for any substitutions of the form (for any $L>K$):
\begin{equation}
\label{eq:KUnifSub}
P(X_{K:L}\!=\!x_{K:L} | r x_{0:K}) = P(X_{0:L\!-\!K}\!=\!x_{K:L} | \mathcal{P}_K(r, x_{0:K})),
\end{equation}
since if this was not the case then $\morph$ would not consistently assign the correct future morph for some memory states.

For notational brevity,
 we recursively define the set of functions $\{r_n:\mathcal{R}\times\mathcal{X}^{nK} \to \mathcal{R}\}_n$ as:
\begin{equation}
\label{eq:PManyTimes}
r_n\!\left(r, x_{0:nK}\right) = \mathcal{P}_K\!\left(r_{n\!-\!1}\!\left(r, x_{0:(n\!-\!1)K}\right), x_{(n\!-\!1)K:nK}\right)
\end{equation}
for $n\geq1$, and $r_0 = r$.
For any $M\in\ints^+$, we can thus expand the probability distribution
\begin{align}
P (X_{0:MK} \!=\! x_{0:MK} \,|\, r) \hspace{-9em} & \hspace{9em}  \nonumber \\
& = P (X_{0:K} \!=\! x_{0:K} \,|\, r)  \, P (X_{K:KM} \!=\! x_{K:MK} \,|\, r x_{0:K} ) \nonumber \\
& = P (X_{0:K} \!=\! x_{0:K} \,|\, r_0) \, P (X_{0:(M\!-\!1)K} \!=\! x_{K:MK} \,|\, r_1\!\left(r, x_{0:K}\right) )\nonumber \\
& = P (X_{0:K} \!=\! x_{0:K} \,|\, r_0) \, P (X_{0:K} \!=\! x_{K:2K} \,|\, r_1\!\left(r, x_{0:K}\right) )\nonumber \\
& \quad \cdot P (X_{K:(M\!-\!1)K}\!=\!x_{2K:MK} \,|\, r_1\!\left(r, x_{0:K}\right) x_{K:2K}) \nonumber \\
& = P (X_{0:K} \!=\! x_{0:K} \,|\, r_0) \, P (X_{0:K} \!=\! x_{K:2K} \,|\, r_1\!\left(r, x_{0:K}\right) )\nonumber \\
& \quad \cdot P (X_{0:(M\!-\!2)K}\!=\!x_{2K:MK} \,|\, r_2\!\left(r, x_{0:2K}\right)) \nonumber \\
& = \ldots \nonumber \\
& = \prod_{i=0}^{M-1} P (X_{0:K} \!=\! x_{iK:(i+1)K} \,|\, r_i(r, x_{0:iK})).
\end{align}
Here, Bayesian expansion allows us to make the first and third equalities,
 and $K$-unifilarity (via Eqs.~\ref{eq:KUnifSub} and \ref{eq:PManyTimes}) allows us to make the substitutions for the second and fourth equalities.
Thus, we can use this to generate a future morph for words of any length $KM$, as a function of the probability distribution over the next $K$ symbols.
(Probabilities over words of length that are not multiples of $K$ can always be found by taking marginals of a longer word that is a multiple of $K$.)

The contrapositive implies that if $\mathcal{F}$ does not assign probabilities in this way,
 then it cannot satisfy Eq.~\ref{eq:KUnifSub},
 and hence is not a $K$-unifilar pre-model.
\end{proof}
\end{lemma}
An immediate corollary of this is that a $1$-unifilar pre-model is exactly a unifilar hidden Markov model.
It also immediately follows that any $K$-unifilar pre-model for $K \in \ints^+$ is also an $NK$-unifilar pre-model for $N\in\ints^+$ -- and particularly, unifilar models are $K$-unifilar pre-models for all $K\in\ints^+$.
For any maximum dimension $d$, this sets up the hierarchy $\mathcal{M}^d_1 \subseteq \mathcal{M}^d_K \subseteq \mathcal{M}^d_\infty$.

\subsection{Calculating the distortion $D_e$}
\label{sm:dist}
For the purposes of this article, we are interested in minimizing the distortion (Eq.~\ref{eq:Distortion}) of models.
Particularly, consider two processes $P$ and $\hat{P}$,
 and for calculational convenience, define the arguments within the limit of Eq.~\ref{eq:Dfutpast} and  Eq.~\ref{eq:Distortion} as
\begin{align}
	\label{eq:dLDef}
	D_{L} (P,\hat{P}) & = \sum_{\past{x}} P(\past{x}) \,  D_L(P, \hat{P} | \past{x}),
\end{align}
such that (Eq.~\ref{eq:Distortion}) may be written
\begin{equation}
D_e(P, \hat{P}) = \lim_{L\to\infty} D_L(P,\hat{P}).
\end{equation}

\begin{lemma}\label{lem:stationaryDist}
For stationary processes $P$ and $\hat{P}$:
 \begin{equation}
	 D_e(P,\hat{P}) = D_{K}(P,\hat{P}) = D_{1}(P,\hat{P}) \qquad \forall K\in\ints^+.
 \end{equation}
\begin{proof}
	According to Bayesian rules, for $L>1$:
	\begin{equation}
		P(x_{0:L}|\past{x}) = P(x_{0}|\past{x}) P(x_{1:L} | \past{x} x_0).
	\end{equation}
Then (expanding the definition of KL-divergence),
\begin{align}
	\label{eq:KLStationary}
LD_L(P,\hat{P}) \hspace{-3em}&\hspace{3em}
		= \sum_{\past{x}}P(\past{x}) \sum_{x_{0:L}}P(x_{0:L} |\past{x})\log \frac{P(x_{0:L}|\past{x})}{\hat{P}(x_{0:L} |\past{x})} \nonumber\\
		& = \sum_{\past{x}}P(\past{x}) \sum_{x_0}P(x_0 |\past{x})\log \frac{P(x_0|\past{x})}{\hat{P}(x_0 |\past{x})} \nonumber\\
		&\qquad +  \sum_{\past{x}}P(\past{x})\sum_{x_{0:L}}P(x_{0:L} | \past{x}) \log \frac{P(x_{1:L}|\past{x}x_0)}{\hat{P}(x_{1:L} | \past{x}x_0)}.
\end{align}
Since $P$ and $\hat{P}$ are stationary, we can substitute $P(x_{1:L} | \past{x} x_0) = P(x_{0:L-1} | \past{x} )$, and hence for $L>1$:
\begin{equation}
LD_L(P,\hat{P}) = D_{1}(P,\hat{P}) + (L-1) D_{L-1}(P,\hat{P}).
\end{equation}
Then, inductively
\begin{equation}
	D_L(P,\hat{P}) = D_1(P,\hat{P}) \qquad \forall L\geq 1,
\end{equation}
and hence also the limit
\begin{equation}
	D_e(P,\hat{P}) = \lim_{L\to\infty} D_L(P,\hat{P}) = D_1(P,\hat{P}).
\end{equation}
\end{proof}
\end{lemma}

Similarly, this enables a shortcut to calculate the distortion between a process and a $K$-unifilar pre-model:
\begin{lemma}
\label{lem:MarkovOrderDistortion}
	Let $P$ be some stationary process, and $(\encod,\morph_K,\mathcal{P}_K)$ be a $K$-unifilar pre-model whose future morph is generated by $\hat{P}(X_{0:K} | \encod(\past{x}))$.
	Then,
	\begin{equation}
		D_e(P,\hat{P}) = D_{K}(P,\hat{P}).
	\end{equation}
\begin{proof}
We can group each word of $K$ contiguous symbols from process $P$ over alphabet $\mathcal{X}$ into a single symbol of process $Q$ with alphabet $\mathcal{Y} := \mathcal{X}^K$, such that the distributions over the two processes are related by
\begin{equation}
Q(\cdots,y_0,y_1,\cdots) = P(\cdots,x_{0:R},x_{R:2R},\cdots).
\end{equation}
We similarly group $\hat{P}$ into $\hat{Q}$.
Since $(\encod,\mathcal{F},\mathcal{P}_K)$ is a $K$-unifilar pre-model for $\hat{P}$,
 it immediately defines a unifilar model for $\hat{Q}$ with encoding $\mathcal{E}_Q : \past{Y} \to \mathcal{R}$
 satisfying $\mathcal{E}_Q(\past{y}) = \mathcal{E}(\past{x})$ where $\past{y}$ is the grouping of $\past{x}$.
Such combination has no effect on the KL--divergence and hence the distortion satisfies
\begin{equation}
D_e(Q,\hat{Q}) = K D_e(P,\hat{P}).
\end{equation}
Meanwhile, using Lemma~\ref{lem:stationaryDist}):
\begin{align}
D_e(Q,\hat{Q}) & = D_1(Q,\hat{Q}) \nonumber \\
  & = \sum_{\past{y}} Q(\past{y}) D_1\!\left( Q, \hat{Q} | \past{y} \right) \nonumber \\
  & = \sum_{\past{x}} P(\past{x}) K D_{K}\!\left( P, \hat{P} |\past{x} \right) \nonumber \\
  & = K D_K(P,\hat{P}).
\end{align}
And hence $D_e(P,\hat{P}) =  D_K(P,\hat{P})$.
\end{proof}
\end{lemma}

\subsection{Bounding the minimum distortion}

A consequence of the above lemmas is to make it easier to calculate bounds that minimize $K$-unifilar pre-models and unifilar models:
\begin{lemma}
\label{lem:ChangingGoalposts}
Let $P$ be a stationary process, and let $K, N \in \ints^+$.
Then for every dimension $d$
\begin{equation}
\label{eq:ChangingGoalposts}
\min_{M \in \mathcal{M}^d_{K}} D_e(P, M) \geq \min_{M \in \mathcal{M}^d_{KN}} D_K(P, M).
\end{equation}
\begin{proof}
First, Lemma~\ref{lem:MarkovOrderDistortion} tells us that for any $K$-unifilar pre-model $D_e(P,M) = D_K(P,M)$
 and hence $\min_{M \in \mathcal{M}_{K}} D_e(P, M) = \min_{M \in \mathcal{M}_{K}} D_K(P, M)$.
Next, the hierarchy of models $\mathcal{M}^d_K \subseteq \mathcal{M}^d_{KN}$,
 and hence a minimization over $\mathcal{M}^d_{KN}$ lower bounds a minimization over $\mathcal{M}^d_{K}$.
This proves Eq.~\ref{eq:ChangingGoalposts}.
\end{proof}
\end{lemma}

Why might we choose some finite $K$ rather than minimizing $D_e$ over generic ($K=\infty$) pre-models?
First, $\mathcal{M}_\infty$ is too permissive, and $D_e$ is too forgiving.
Consider any (stationary, aperiodic) process,
 and let us assign a $1$--dimensional model, but whose single future morph defines probabilities arbitrarily far into the future, with statistics that increasingly match those of $P$ (e.g.\ by assigning a weighted average of $P$'s causal state's future morphs).
For such a (highly non-unfilar) model, $D_e \to 0$, saturating the arithmetically obvious bound.
Second, the parameter space of $P\left(\future{X}|\mathcal{E}(\past{x})\right)$ is a priori infinite,
 and the above Lemmas allow us to greatly reduce the relevant parameter space of future morphs in our calculations to probability distributions over words of finite length $K$.

Although the above Lemmas allow us to reduce the dimensionality of $\morph$ assigned to each memory state,
 there still remains the problem that there are infinitely many possible partitions of the past $\encod$.
We will argue, with the next set of Lemmas, that for exploring minimum bounds, we can restrict ourselves to the finite set of coarse--grainings over the causal states.

First, we note that the problem of finding a minimum distortion model is trivial if we allow models with high--enough dimension:
\begin{lemma}
\label{lem:EpsMinDist}
For a process $P$ with topological complexity $d_c$,
 then for all model dimensions $\hat{d}_c \geq d_c$,
 there is a zero--distortion unifilar model $M$ such that
\begin{equation}
D_K(P, M) = 0 \qquad \forall K\in\ints^+.
\end{equation}
\begin{proof}
Such a minimum distortion model can be realized by the $\varepsilon$--machine of $P$, which has zero distortion.
By definition of causal states, $P(X_{0:K}|\past{x}) = P(X_{0:K}|\varepsilon(\past{x}))$ for all $K$ and $\past{x}$,
 and hence $\mathcal{D}_{\rm KL}(P(X_{0:K}|\past{x}) \, || \, P(X_{0:K}|\varepsilon(\past{x}))) = 0$ for all $L$ and $\past{x}$.
This implies that all $D_K(P, \hat{P}) = 0$.
\end{proof}
\end{lemma}

\begin{lemma}
\label{lem:merge}
For any process $P$, for every dimension $d$,
	and for every $K\in\ints^+$,
	if two pasts $\past{y}$ and $\past{z}$ have identical future morphs in $P$,
	then there exists a pre-model $\left(\encod^d_{\min}, \morph\right)$
	that minimizes $D_K$ and satisfies
	\begin{equation}
	\label{eq:SameMorphsOpt}
		\encod^d_{\min}(\past{y}) = \encod^d_{\min}(\past{z}).
	\end{equation}
\begin{proof}
If $d \geq d_c$, the $\varepsilon$--machine has the minimum distortion (see Lemma~\ref{lem:EpsMinDist}) and the causal state partition $\varepsilon$ satisfies Eq.~\ref{eq:SameMorphsOpt} by definition.
Likewise, if $d_c=1$, there is only one encoding onto a single state and so Eq.~\ref{eq:SameMorphsOpt} is trivially satisfied.
Thus, it remains only to prove the cases where $1 < d < d_c$.

Suppose $\past{y}$ and $\past{z}$ have identical future morphs in $P$,
 and we have a pre-model $(\encod, \morph)$ where encoding $\encod(\past{y}) \neq \encod(\past{z})$.
Then without loss of generality (by switching the labels of $\past{y}$ and $\past{z}$ if necessary) the implied statistics $\hat{P}$ satisfy
\begin{equation}
	D_K(P,\hat{P} | \past{y}) \leq 	D_K(P, \hat{P} | \past{z}),
\end{equation}
where $	D_K(P, \hat{P} | \past{x} )$ is defined in Eq.~\ref{eq:dLDef}.

We can then construct a new pre-model with identical $\morph$,
 but with a new mapping $\encod^{'}$ identical to $\encod$ in every way, except it now maps $\past{z}$ to $\encod\!\left(\past{y}\right)$ instead of $\encod\!\left(\past{z}\right)$.
This has implied statistics $\hat{P}'$, and
\begin{equation}
	D_K(P, \hat{P}' | \past{z}) = D_K(P, \hat{P} | \past{y}) \leq D_K(P, \hat{P} | \past{z}).
\end{equation}
Hence $D_K$ of $(\encod', \morph)$ satisfies
\begin{align}
D_K(P, \hat{P}') \hspace{-3em} & \nonumber \\
	& = P(\past{z}) D_K(P, \hat{P}' | \past{z} ) + \sum_{\past{x} \neq \past{z}} P(\past{x}) D_K(P, \hat{P}' | \past{x}) \nonumber\\
	& = P(\past{z}) D_K(P, \hat{P}' | \past{z} ) + \sum_{\past{x} \neq \past{z}} P(\past{x}) D_K(P, \hat{P} | \past{x}) \nonumber\\
	& \leq P(\past{z}) D_K(P, \hat{P} | \past{z} ) + \sum_{\past{x} \neq \past{z}} P(\past{x}) D_K(P, \hat{P} | \past{x}) \nonumber\\
	& = D_K(P, \hat{P}).
\end{align}
It then follows that for any pre-model where $\past{y}$ and $\past{z}$ with identical future morphs map to different memory states,
 there is another pre-model of the same (or lower) dimension such that $\past{y}$ and $\past{z}$ map to the same memory state,
  and this pre-model has the same or lower distortion.
Hence, among the minimum distortion pre-models of a given process, there will always be a pre-model where $\past{y}$ and $\past{z}$ map to the same memory state.
\end{proof}
\end{lemma}

The above lemma implies the following:
\begin{lemma}
\label{lem:MergeCausal}
For a process $P$ with topological complexity $d_c$
For every dimension $d < d_c$,
 and every $K\in\ints^+$
there is a pre-model $M$ that minimizes $D_K(P,M)$ whose encoding $\encod^d_{\min}$ is a map onto a coarse graining of the causal states $\mathcal{S}$.
\begin{proof}
This follows from Lemma~\ref{lem:merge} by noting that pasts in the same causal state have the same future morph by definition.
\end{proof}
\end{lemma}

Such coarse--grained models admit a computational shortcut for calculating their distortion:
 we can define $\pi_i := \sum_{\past{x} \in s_i} P(\past{x})$ for each causal state $s_i\in\mathcal{S}$,
 such that the distortion is
\begin{equation}
	D_K(P, \hat{P}) = \sum_{i=1}^{d_c} \pi_i \, D_K(P,\hat{P} | s_i),
\end{equation}
where for each $s_i$, $D_K(P,\hat{P} | s_i) := D_K(P, \hat{P} | \past{x})$ for one arbitrary choice of $\past{x}\in s_i$ (since by definition the value is the same for all such choices).

Crucially, the above Lemma enables an exhaustive search for $D_K$--optimal pre-models for any process with finite topological complexity -- instead of having to consider an infinite number of possible partitions, we can iterate through the various combinations of causal states.

\begin{lemma}
\label{lem:LeaveItAlone}
For a process $P$,
 for every dimension $d$,
 and every $K\in\ints^+$
 there is pre-model $M = (\encod^d_{\min}, \morph^d_{\min}) \in M^d_\infty$ that minimizes $D_K(P,M)$
 such that every memory state $r_i$ whose entire pre-image (except for a measure zero subset with respect to $P(\past{x})$) is mapped to the causal state $s_{i'}$,
 has the future morph $\morph^d_{\min}(r_i) = P(\future{X}|s_{i'})$.
\begin{proof}
Again, we consider the case where model dimension $d<d_c$ since the case $d\geq d_c$ is trivially satisfied by the $\varepsilon$--machine (Lemma~\ref{lem:EpsMinDist}).
From Lemma~\ref{lem:MergeCausal} we may restrict ourselves to coarse-grainings of $\varepsilon$--machines.

Let $s_i$ be the causal states of the process, indexed such that $i=1\ldots k$ correspond to states that are not merged, and $i=k+1\ldots d_c$ are merged into states $r_1 \ldots r_{d-k}$ (i.e.\ the states of the model are $\mathcal{R} = \{s_1,\ldots s_k, r_1, \ldots r_{d-k}\}$).
As the Lemma's claim is vacuously true if $k=0$, we consider the cases when $k\geq 1$.
Then, using the coarse--graining structure, and choosing an arbitrary past $\past{x}_i \in s_i$ for each causal state,
\begin{equation}
D_K(P,\hat{P}) = \sum_{i=1}^k \pi_i D_K(P, \hat{P} | \past{x}_i) +  \sum_{i=k+1}^{d} \pi_i D_K(P, \hat{P} | \past{x}_i).
\end{equation}
If for any $i \in [1,k]$ we have $D_K(P, \hat{P} | \past{x}_i) > 0$, we can instead form a new pre-model that assigns the  future morph  $\morph(s_i) = P(\future{X}|s_{i})$.
Then, in this new pre-model $D_k(P, \hat{P}' | \past{x}_i) = 0$ for all such $s_i$.
Thus,
\begin{equation}
D_K(P,\hat{P}') = \sum_{i = k+1}^d \pi_i D_K(P, \hat{P} | s_i) \leq D_K(P,\hat{P}).
\end{equation}
\end{proof}
\end{lemma}
In simpler words: when we look for a $D_K$--optimal pre-model by merging states of an $\varepsilon$--machine, we should not alter the future morphs of any ``unmerged'' states from how they are in the $\varepsilon$--machine.

\begin{theorem}
\label{thm:BoundOnUnifilar}
Let $P$ be any process.
Then, for every model dimension $d$, the algorithm in Box.2 produces a bound on the minimum distortion $D_e$ of a unifilar model by considering all possible mergers of  causal states, and searching over the next $K$ output statistics.
\begin{proof}
Lemma~\ref{lem:MergeCausal} tells us that a $D_K$--optimal pre-model exists that is formed by merging causal states,
 but says nothing as to whether such a pre-models is unifilar.
However, Lemma~\ref{lem:ChangingGoalposts} implies that such a $D_K$-optimized pre-model bounds the $D_e$--optimized unifilar models.
Moreover, $D_K$, as a function of $\hat{P}(X_{0:K} | \encod(\past{x}))$, depends only on the next $K$ output statistics.
As such, to find this minimum bound, we can iterate through the various coarse--grainings of casual states of a given dimension,
 and search through finite parameter space of $\hat{P}$ to minimize $D_K$.
The minimum found here will lower bound the lowest value of $D_e$ achievable by a unifilar model.
The proof is illustrated in~\fref{fig: illusbound}.
\begin{figure}[htbp]
	\centering
	\includegraphics[width = 0.9\linewidth]{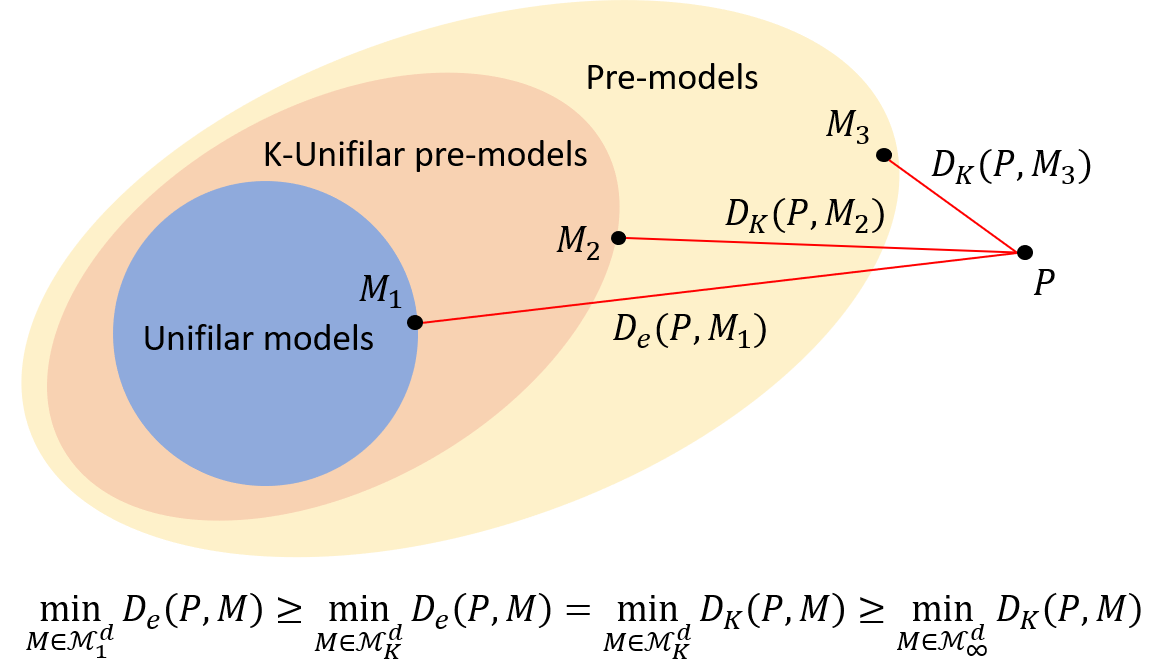}
	\caption{$P$ represents the given stochastic process.  For a given model dimension $d$, $\mathcal{M}_\infty^d$ represent all pre-models,  $\mathcal{M}_K^d$ represents all K-unifilar pre-models, and $\mathcal{M}_1^d$ represents all unifilar models.}
	\label{fig: illusbound}
\end{figure}
 \end{proof}
\end{theorem}

\subsection{Minimum--distortion approximations of Markov processes}
For Markov processes, this bound is tight.
First:
\begin{lemma}
\label{lem:MergeRUnif}
Let $P$ be a process with Markov order $\mr$.
Then, any encoding map $\encod$ formed by coarse--graining the causal states of $P$ admits an $\mr$-unifilar pre-model.
An $\mr$-unifilar pre-model that minimizes $D_{\mr}(P,\hat{P})$ is in this set.
\begin{proof}
Let $\varepsilon$ be the encoding map onto causal states, and $\mathcal{C}: \mathcal{S} \to \mathcal{R}$ be the coarse--graining map from causal states $\mathcal{S}$ to memory states $\mathcal{R}$, such that $\encod = \mathcal{C} \cdot \varepsilon$.
We must demonstrate that there exists a $\mathcal{P}_\mr$ such that $\mathcal{P}_{\mr} \left( \encod\!\left(\past{x}\right), x_{0:\mr} \right) = \encod\!\left(\past{x} x_{0:\mr}\right)$ for all $\past{x}$ and all $x_{0:\mr}$.
Since $P$ has Markov order $\mr$, there exists an $\tilde{\varepsilon}: \mathcal{X}^{\otimes \mr} \to \mathcal{S}$ that acts on only the final $\mr$ symbols of the history to identify the causal states, and for every $\past{x}$ that ends in the same $x_{-\mr:0}$, $\varepsilon(\past{x}) = \tilde{\varepsilon}(x_{-\mr:0})$.
We thus have an $\tilde{\encod} = \mathcal{C}\cdot\tilde{\varepsilon}$ that maps to the same memory state as $\encod$ for every $x_{-\mr:0}$ and every $\past{x}$ that ends in $x_{0:\mr}$.

Thus, now consider $\mathcal{P}_{\mr} \left( \encod\!\left(\past{x}\right), x_{0:\mr} \right) = \encod\!\left(\past{x} x_{0:\mr}\right) = \tilde{\encod}(x_{0:\mr})$.
A candidate $\mathcal{P}_{\mr}$ can be defined such that it is completely independent of its first argument (the current machine state $r\in\mathcal{R}$) and instead takes exactly the same value as $\tilde{\encod}$ applied to its second input (the recently output word $x_{0:\mr} \in \mathcal{X}^\mr$).
Thus, there exists an $\mr$-unifilar pre-model for any such $\encod$.

Now, from Lemma~\ref{lem:MergeCausal},
 we know that the $D_\mr$ optimal ($\infty$-unifilar) pre-model $(\encod,\morph)$ is formed by merging causal states,
 but a priori we have not demonstrated that such a model is $\mr$-unifilar.
To calculate $D_\mr$,
 we need only evaluate the probabilities associated with the first $\mr$ outputs.
Thus, we can define $\mathcal{F}_\mr$ as per Lemma~\ref{lem:KUnifWordK} such that $\mathcal{F}_\mr$ and $\mathcal{F}$ perfectly agree on the statistics of the first word of length $\mr$ in the morph.
Then, taking the encoding map of $\encod$ from this pre-model, we form a $\mr$-unifilar model $(\encod, \morph_\mr, \mathcal{P}_\mr)$ where $\mathcal{P}_\mr$ is defined as above.
Since $D_\mr$ only depends on the first $\mr$ steps of the future, minimizing the value of $D_\mr(P,\hat{P})$ provides the optimal $\mr$-unifilar pre-model.
As such, we have formed a minimum--distortion $\mr$-unifilar pre-model.
\end{proof}
\end{lemma}

When the process is Markovian, these lemmas then give us a systematic method for finding the minimum distortion unifilar hidden Markov model.
\begin{theorem}
\label{thm:MarkovTight}
Let $P$ be a Markov process.
Then, for every model dimension $d$, we can find a minimum distortion unifilar model by merging causal states, and searching over the next--output statistics.
\end{theorem}
\begin{proof}
If $d\geq d_c$, this encoding is the $\varepsilon$-machine (Lemma~\ref{lem:EpsMinDist}).
In the case $d<d_c$, we use Lemma~\ref{lem:MergeCausal}, to find a $\encod^d_{\min}$ that is a coarse--graining of causal states.
Then Lemma~\ref{lem:MergeRUnif} specialized to $\kappa=1$ implies that this model can be made unifilar.
\end{proof}

\section{Classical minimum distortion models: Examples}
\label{sm:boundseg}
As illustrative examples, we present the minimum--distortion approximate classical model for asymmetric process with $p=0.3,q=0.8$ in~\fref{fig:approxCA} and for quasi cycle with $p = 0.5 ,\delta = 0.1$ in~\fref{fig:approxQC}.

\begin{figure}[htbp]
	\centering
	\includegraphics[width=0.65\linewidth]{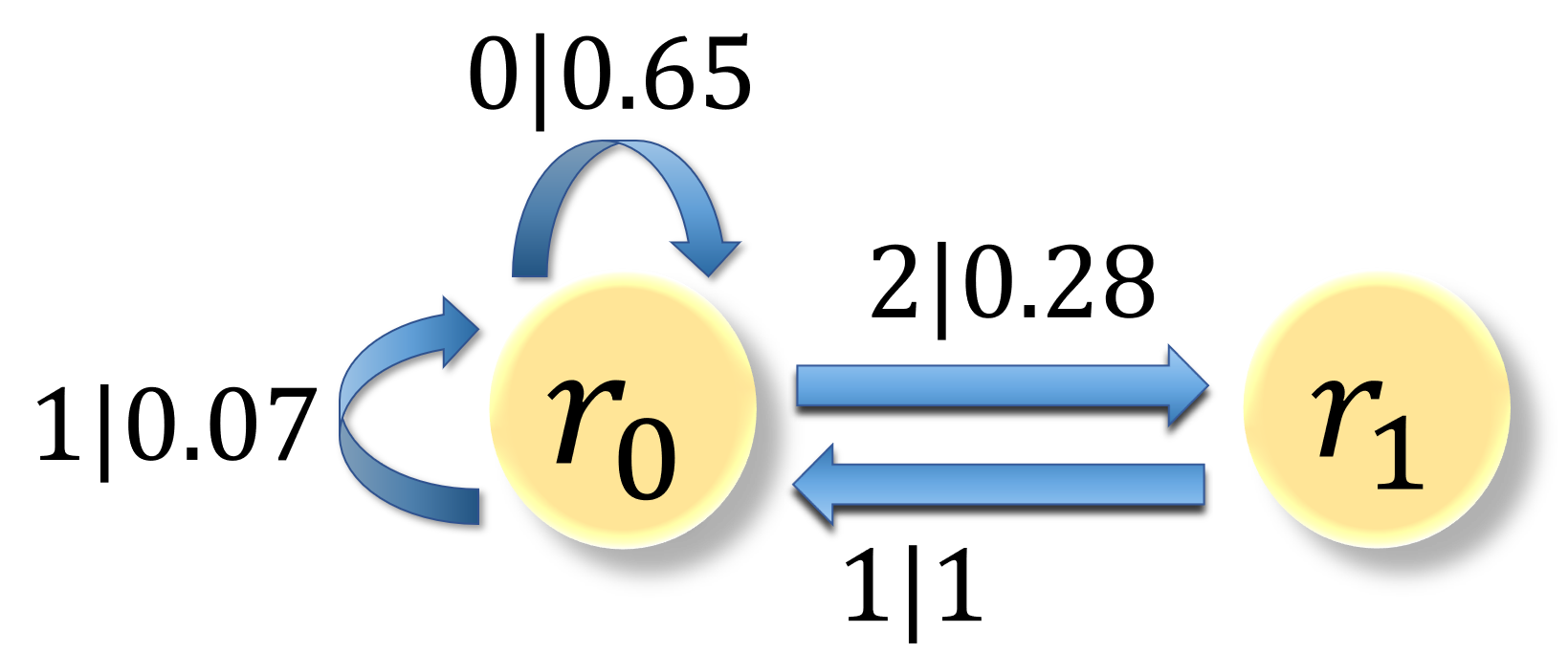}
	\caption{\caphead{Approximate model for asymmetric process.}
	The approximate classical model for the asymmetric example with $p=0.3$ and $q = 0.8$.
	$s_0$ and $s_1$ of the original asymmetric process are merged into $r_0$, with transition probabilities identified by a minimization. $s_2$ is mapped to $r_1$ with an unchanged future morph.}
	\label{fig:approxCA}
\end{figure}

\begin{figure}[htbp]
	\centering
	\includegraphics[width=0.65\linewidth]{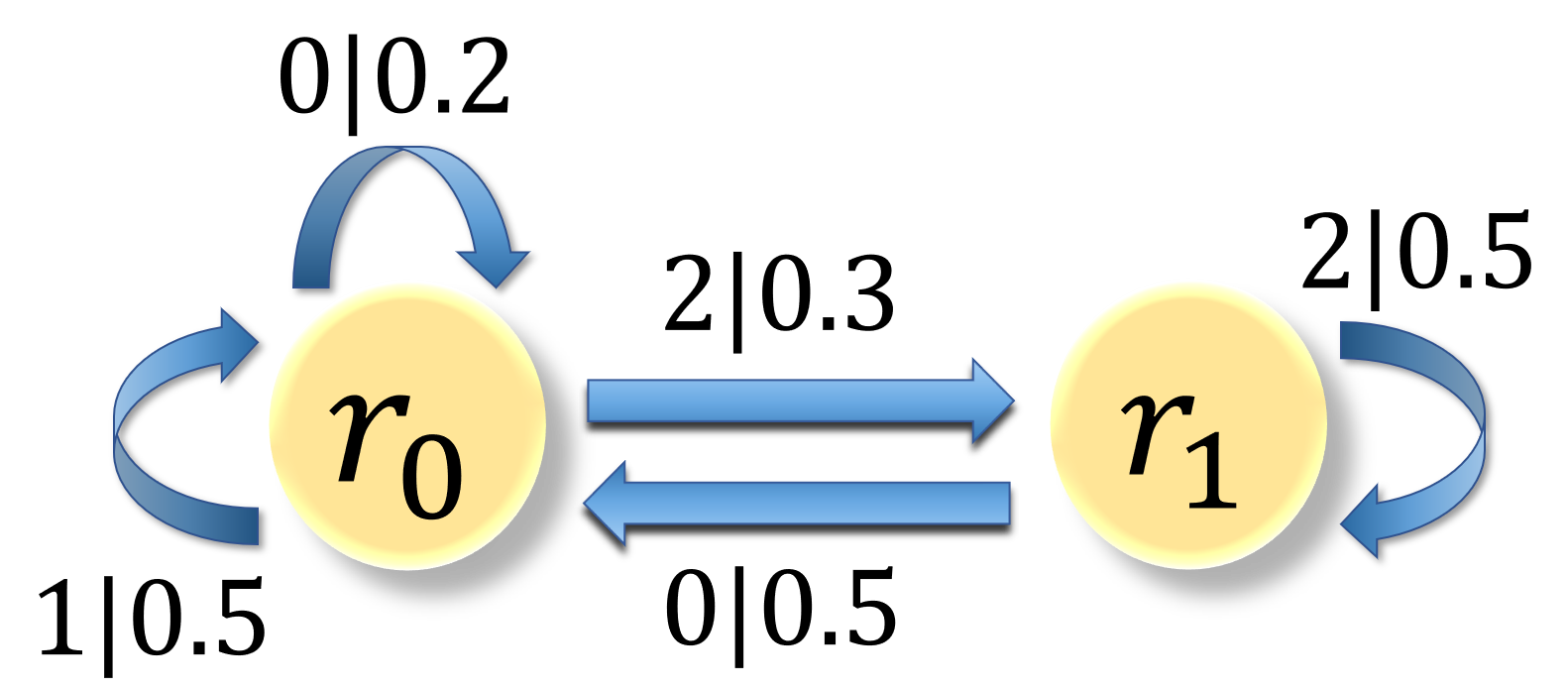}
	\caption{\caphead{Approximate model for quasicycle.}
	The approximate classical model for the quasicycle example with $p=0.5$ and $\delta = 0.1$.
	$s_0$ and $s_1$ of the original quasicycle process are merged into $r_0$.
	The transition probability is minimized by going through the probability
	vector space. $s_2$ is mapped to $r_1$ with an unchanged future morph.}
	\label{fig:approxQC}
\end{figure}

For the third example, we evaluate $D_N$ as a lower bound on the error of the optimal classical model for any discrete renewal process with $N$ states, as shown in~\fref{fig:ex3error}.

\clearpage
\section{Details of experimental quantum probability distribution.}
\label{app:IBM}
\fref{fig:IBMprob} compares the conditional probability distribution realized on ``ibmq athens" with the corresponding noiseless distribution of the learned quantum model.
\begin{figure}[htbp]
	\centering
	\includegraphics[width=0.6\linewidth]{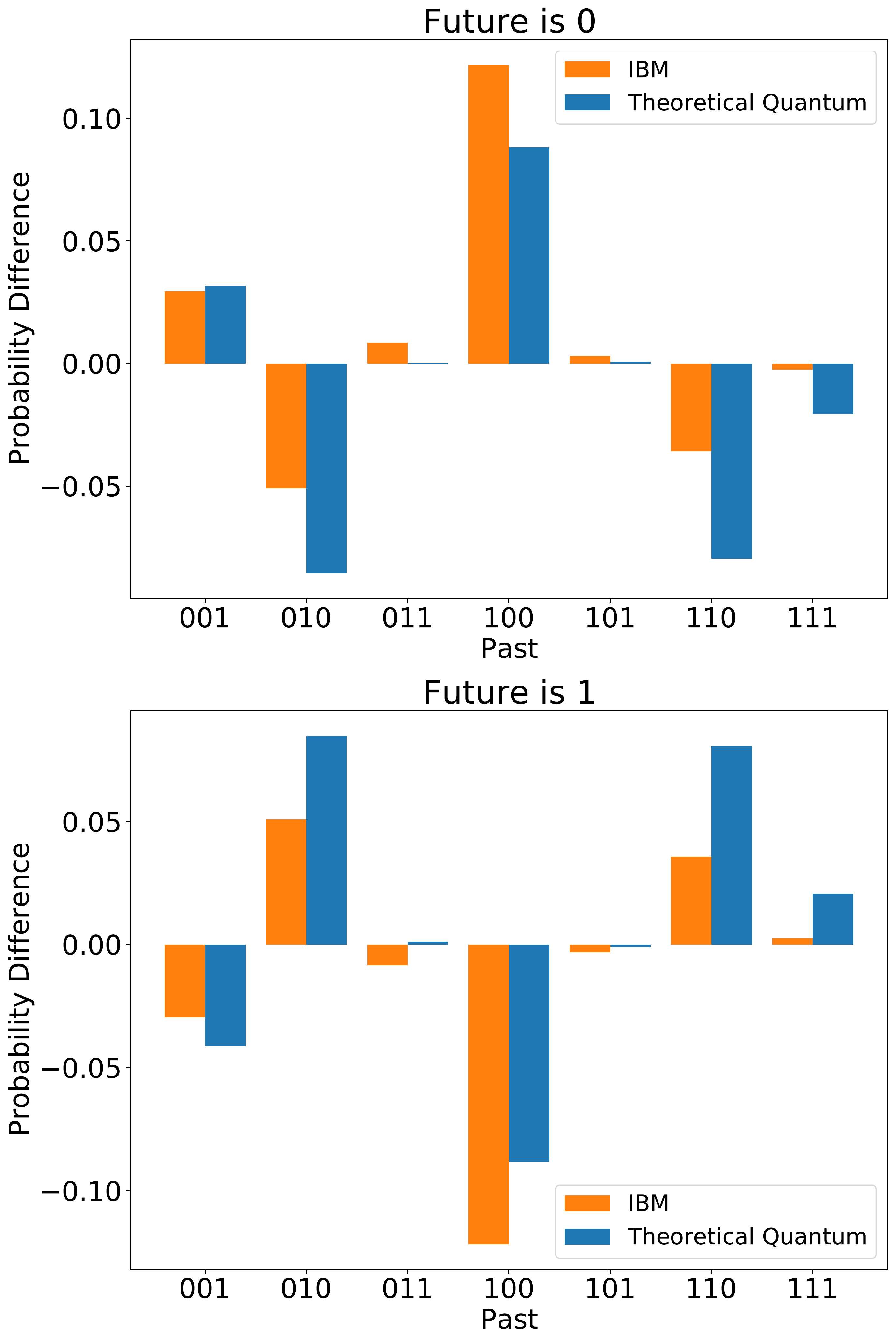}
	\caption{
	\caphead{Comparison of experimental (orange) and theoretical (blue) realizations of the learned quantum model to the true process. }
	Each datum represents the difference between the model's conditional probability for a given past, and the true process's conditional probability.
	}
	\label{fig:IBMprob}
\end{figure}
\clearpage

\end{document}